\newcommand\SANSCOMMENTAIRE[1]{}

\newcommand\shorter[2][]{}

\newcommand\shortermcu[1]{}

\newcommand\manondufutur[2][]{\SANSCOMMENTAIRE{\todo[inline,color=pink!80!white,caption={2do},#1]{\begin{minipage}{\textwidth-4pt}
Pour la Manon du Futur:			#2\end{minipage}}}}

\newcommand{\olivier}[2][]{\SANSCOMMENTAIRE{\todo[inline,color=blue!40,caption={2do}, #1]{\begin{minipage}{\textwidth-4pt}
Olivier:			#2\end{minipage}}}}

\newcommand{\manon}[2][]{\SANSCOMMENTAIRE{\todo[inline,color=violet!20!white,caption={2do},#1]{\begin{minipage}{\textwidth-4pt}
Manon:			#2\end{minipage}}}}

\newcommand{\olivierpourmanon}[2][]{\SANSCOMMENTAIRE{\todo[inline,color=blue!40,caption={2do}, #1]{\begin{minipage}{\textwidth-4pt}
Olivier $\to$ Manon:			#2\end{minipage}}}}
\newcommand{\olivierpourlecteur}[2][]{\SANSCOMMENTAIRE{\todo[inline,color=blue!40,caption={2do}, #1]{\begin{minipage}{\textwidth-4pt}
Olivier $\to$ Arnaud:			#2\end{minipage}}}}

\newcommand\Lesfonctionsquoi{\mathcal{F}}
\newcommand\manonclass{\linearderivlength^\bullet}
\newcommand\manonclasslightidealsigmoid{\linearderivlength^\oslash}
\SANSCOMMENTAIRE{\newcommand\manonclasslighttanh{\linearderivlength^\circ}}
\newcommand\MANONlim{ELim}
\newcommand\MANONlimd{E_{2}Lim}

\newcommand\manonclasslightidealsigmoidlim{\bar{\manonclasslightidealsigmoid}}
\newcommand\polynomialb{ $\bar{\fonction{cond}}$-polynomial}
\newcommand{\signb}[1]{\bar{\fonction{cond}}(#1)}
\newcommand{\signbb}[1]{\bar{\fonction{cond}}\left(#1\right)}
\newcommand\unaire[1]{with respect to the value of $#1$}
\newcommand\Encode{\textit{Encode}}
\newcommand\EncodeMul{\textit{EncodeMul}}
\newcommand\Decode{\textit{Decode}}
\newcommand\DP{\operatorname{DP}}
\newcommand\If{\operatorname{If}}
\newcommand\send{\operatorname{send}}
\newcommand\sendsymbol{\mapsto}

\newcommand\motnouv[1]{\emph{#1}}

\newcommand\MYVEC[1]{\vec{#1}}

\documentclass[a4paper]{article}
\usepackage{mathptmx}
\usepackage[utf8]{inputenc}
\usepackage[T1]{fontenc}
\usepackage[final]{microtype}
\usepackage[english]{babel}
\usepackage{csquotes}
\usepackage{amsmath}
\usepackage{amsthm}
\usepackage{amsfonts}
\usepackage{commath}
\usepackage{bm}
\usepackage{xfrac}
\usepackage{mathtools}
\usepackage{authblk}
\usepackage{url}
\usepackage[pdfencoding=auto,psdextra,hidelinks]{hyperref}
\usepackage{bookmark}
\usepackage{float}
\usepackage{booktabs}
\usepackage{subcaption}
\usepackage{tikz}
\usetikzlibrary{calc,positioning,shapes,arrows}

\usepackage{color}
\usepackage{todonotes}

\newcommand\THEPAPIERS{\cite{MFCS2019,MFCSJournalProject}}
\newcommand\THEPAPIERSPLUS{\cite{MFCS2019,MFCSJournalProject,BlancBournezMCU22vantardise}}
\usepackage{mathdots}

\usepackage{versions}
\excludeversion{pascettesection}

\newcommand\dyadic{\mathbb{D}}

\newcommand\vectorl[1]{{\mathbf#1}}

\newcommand\vn{\vectorl{n}}

 \newcommand\R{\mathbb{R}}
\newcommand\Q{\mathbb{Q}}
\newcommand\N{\mathbb{N}}
 \newcommand\Z{\mathbb{Z}}
 
 \newcommand\RR{\R}

\newcommand\Greg{Grzegorczyk}

\newcommand\BRN{{\rm BRN}}

\newcommand\ODE{{\rm ODE}}

\newcommand\LI{{\rm LI}}
\newcommand\Elem{\mathcal{E}}
\newcommand\Gregn{\mathcal{E}_n}
\newcommand\PR{\mathcal{P}\mathcal{R}}

\newcommand{\fonction}[1]{\textrm{#1}}

\newcommand\projection[2]{\mathbf{\pi}_{#1}^{#2}}
\newcommand{\sucs}{\mathbf{s}}

\newcommand\plus{\mathbf{+}}
\newcommand\minus{\mathbf{-}}
\newcommand\gE{\mathbf{E}}

\newcommand\dint[4]{\int_{#1}^{#2}{#3}{\delta #4}}

\newcommand\fallingexp[1]{\overline{2}^{#1}}

\newcommand{\succun}[1]{\mathbf{1}({#1})}
\newcommand{\succzero}[1]{\mathbf{0}({#1})}

\newcommand{\zero}{\mathbf{0}}

\newcommand{\sign}[1]{\fonction{sg}(#1)}

\newcommand{\cond}[3]{\fonction{if}(#1,#2,#3)}

\newcommand{\tu}[1]{\mathbf{#1}}

\newcommand{\cp}[1]{\mathbf{#1}}
\newcommand{\Ptime}{\cp{PTIME}} 
\newcommand{\classP}{\cp{P}}      

\newcommand{\NP}{\cp{NP}}

\newcommand{\FPtime}{\cp{FPTIME}}
\newcommand{\FP}{\cp{FP}}

\newcommand{\Pspace}{\cp{PSPACE}}

\newcommand{\FPspaceN}{\cp{FPSPACE}}

\newcommand{\suffix}{\textsf{suffix}}

\newcommand{\approximate}[2][none]{%
 \tikz[baseline=(a.base)]\node[draw,rounded corners, inner sep=2pt, outer sep=0pt,fill=#1](a){\ensuremath #2\strut};
 }

\newcommand\danielsigmaold[1]{\sigma(#1)}

\newcommand\danielsigmanameold{\sigma}
\newcommand\danielsigma[1]{\mbox{\approximate{$#1$}}}
\newcommand\danielsigmaeold[2]{\sigma^{#1}{(#2)} }
\newcommand\danielsigmae[2]{\danielsigma{#2}_{#1}}
\newcommand\ltroisname{\approximate{\in\{0,1,2\} }}
\newcommand\ltrois[2]{\ltroisname_{#2}\left(#1\right)}
\newcommand\ldeuxname{\approximate{\in\{0,1\} }}
\newcommand\ldeux[2]{\ldeuxname_{#2}\left(#1\right)}

\renewcommand\mod{\operatorname{mod}}
\newcommand\sig{\operatorname{sig}}

\renewcommand\bar[1]{\overline{#1}}

\newcommand{\dderiv}[2]{\frac{\partial #1}{\partial #2}}
\newcommand{\dderivL}[1]{\frac{\partial #1}{\partial \lengt}}
\newcommand{\dderivl}[1]{\frac{\partial #1}{\partial \ell}}

\newcommand{\lengt}{\mathcal{L}}
 \newcommand\lengthnotation{\ell}
\newcommand{\length}[1]{\mathrm{\lengthnotation}(#1)}
\newcommand{\degre}[1]{\mathrm{deg}(#1)}

\newcommand{\derivlength}{\mathbb{DL}}
\newcommand{\linearderivlength}{\mathbb{LDL}}

\newcommand\base{4}
\newcommand\symboleun{1}
\newcommand\symboledeux{3}
\newcommand\encodageconfiguration{\gamma_{config}}
\newcommand\encodagemot{\gamma_{word}}

\newcommand\Image{\mathcal{I}}

\newtheorem{theorem}{Theorem}
\newtheorem{proposition}{Proposition}
\newtheorem{lemma}{Lemma}
\newtheorem{corollary}{Corollary}

\newtheorem{definition}{Definition}
\newtheorem{remark}{NB}

\usepackage[
backend=biber,
style=alphabetic,
useprefix=false,
url=false,
doi=true,
eprint=false,
maxbibnames=99,
]{biblatex}
\bibliography{bournez,perso}

\title{Polynomial time computable functions over the
	reals characterized using discrete ordinary
	differential equations}
\author{Manon Blanc \footnote{This research was [partially] supported by Labex DigiCosme (project ANR11LABEX0045DIGICOSME) operated by ANR as part of the program « Investissement d'Avenir» Idex ParisSaclay (ANR11IDEX000302).} \\
	LIX, France\\
	manon.blanc@lix.polytechnique.fr\\	
	\and Olivier Bournez \\
	LIX, France \\
	olivier.bournez@lix.polytechnique.fr}
\date{ }

\begin{document}

\maketitle

\begin{abstract}
	The class of functions from the integers to the integers computable in polynomial time has been characterized recently using discrete ordinary differential equations (ODE), also known as finite differences. In the framework of 
	ordinary differential equations, it is very natural to try to extend the approach to classes of functions over the reals, and not only over the integers. Recently, an extension of a previous  characterization was obtained for functions from the integers to the reals, but the method used in the proof,  based on the existence of a continuous function from the integers to a suitable discrete set of reals, cannot  extend to functions from the reals to the reals, as such a function cannot exist for clear topological reasons.
	
	In this article, we prove that it is indeed possible to provide an elegant and simple algebraic characterization of functions from the reals to the reals: we provide a characterization of such functions as the smallest class of functions that contains some basic functions, and that is closed by composition, linear length ODEs, and a natural effective limit schema.
	This is obtained using an alternative proof technique based on the construction of specific suitable functions defined recursively, and  a barycentric method. 
	
	Furthermore, we also extend previous characterizations in several directions: 
	First, we prove that there is no need of multiplication. We prove a normal form theorem, with a nice side effect related to formal neural networks. Indeed, given some fixed error and some polynomial time $t(n)$, our settings produce effectively some neural network that computes the function over its domain with the given precision, for any $t(n)$-polynomial time computable  function $f$.
	
	As far as we know, this is in particular the first time that polynomial time computable functions over the reals are characterized in an algebraic way, in a so simple discrete manner. Furthermore, we believe that no relation between polynomial time computable functions over the reals (in the sense of computable analysis) and formal neural networks has been obtained before. 
	We believe that these characterizations share the above mentioned benefits of these approaches. In particular, compared to already existing characterizations of polynomial time (over the integers or reals), in the so-called field of implicit complexity, this is obtained without any explicit bound on the growth of function, and this is based on a very natural consideration in the framework of differential equations, namely assuming \emph{linear} ordinary differential equations. 
	
	This also points out the possibility of using tools such as changes of variables to capture computability and complexity measures, or as a tool for programming, even over a discrete space. 
	We also believe that these results shed some lights on recent characterizations of complexity classes based on (classical continuous) ordinary differential equations, providing possibly really simpler proof techniques and explanations. 
\end{abstract}



\newpage

\section{Introduction}

Ordinary differential equations can be considered as a universal language for modeling various phenomena in experimental sciences. They have been studied intensively in the last centuries, and their mathematical theory is well-understood: see e.g.  \cite{Arn78,BR89,CL55}. A series of recent articles, initially motivated by understanding analog models of computations has established some characterizations of classical discrete complexity classes from computability theory using ordinary differential equations. In particular, it has been proved that the length of trajectories provides a robust notion of time complexity that corresponds to classical time complexity for models such as Turing machines \cite{ICALP2016ventardise,JournalACM2017}: See \cite{DBLP:journals/corr/BournezGP16} for most recent survey.

Unfortunately, while the above mentioned results are easy to state, their proofs are rather highly technical and mixing considerations about approximations, control of errors, and various constructions to emulate in a continuous fashion some discrete processes. There have been some recent attempts to go to simpler constructions in order to simplify their programming \cite{CIE22}, since these constructions have recently led to a solution of various open problems, with very visible awarded outcomes:  this includes the proof of the existence of a universal ordinary differential equation \cite{ICALP2017}, the proof of the Turing completeness of chemical reactions \cite{CMSB17vantardise}, or hardness of problems related to dynamical systems \cite{GracaZhongHandbook}. 

Initially motivated by trying to go to simpler proofs, several authors have considered their discrete counterparts, that are called discrete ODEs, also known as difference equations \THEPAPIERS. The basic principle is, for a function $\tu f(x)$, to consider its discrete derivative defined as $\Delta \tu f(x)= \tu f(x+1)-\tu
f(x)$. As in these articles, we  intentionally  also write $\tu f^\prime(x)$ for
$\Delta \tu f(x)$ to help to understand
statements with respect to their classical continuous counterparts.  It turns out that this provided some algebraic characterizations of complexity classes, but that a key difference between the two frameworks is that there is no simple expression for the derivative of the composition of functions in the discrete settings, and hence that actually both approaches have some common aspects, but are at the end not yet directly connected. 

Notice that the theory of discrete ordinary differential equations is widely used in some contexts such as function approximation \cite{gelfand1963calcul} and  in \emph{discrete calculus} 
\cite{graham1989concrete,gleich2005finite,izadi2009discrete,urldiscretecalculuslau} for  combinatorial analysis, but is rather unknown. Actually, the  similarities between discrete and continuous statements have  been historically observed, under the terminology of  \emph{umbral} or \emph{symbolic calculus} as early as in the $19$th century, even if not yet fully understood, and often rediscovered in many fields, with various names.


%

Following \cite{MFCS2019}, while the underlying computational content of finite differences theory is clear and has been pointed out many times, no fundamental connections with algorithms and complexity had been formally established before  \THEPAPIERS, where it was proved that many complexity and computability classes  can  be characterized algebraically using discrete ODEs. 



In the context of algebraic classes of functions, a classical notation is the following: Call \emph{operation}  an operator that takes finitely many functions, and returns some new function defined from them. Then $[f_{1}, f_{2}, \dots, f_{k}; op_{1}, op_{2},\dots,op_{\ell}]$ denotes the smallest set of functions containing functions $f_{1}, f_{2}, \dots, f_{k}$ that is closed under operations $op_{1}$, $op_{2}$, \dots $op_{\ell}$. 
Call \emph{discrete function} a function of type $ f: S_{1} \times \dots \times S_{d} \to S'_{1} \times \dots S'_{d'}$, where each $S_{i},S'_{i}$ is either $\N$ or $\Z$.
Write $\FPtime$ for the class of functions computable in polynomial time. 
\olivier{Besoin d'autres classes de complexité?}
A  main result of \THEPAPIERS{} is the following ($\linearderivlength$ stands for linear derivation on length):

\begin{theorem}[\cite{MFCS2019}] \label{th:ptime characterization 2}
	For discrete functions, we have 
	$\linearderivlength= \FPtime$
	where $\linearderivlength =$ $ [\mathbf{0},\mathbf{1},\projection{i}{k}, \length{x}, \plus, \minus, \times, \sign{x} \ ; composition, linear~length~ODE].$
\end{theorem}
That is to say,  $\linearderivlength$ (and hence $\FPtime$ for  functions over the integers) is the smallest class of  functions
that contains   the constant functions $\mathbf{0}$ and $\mathbf{1}$, the projections
$\projection{i}{k}$,  the length function  $\length{x}$ (that maps an integer to the length of its binary representation), 
the addition function $x \plus y$, the subtraction function $x \minus y$, the multiplication function $x\times y$ (that we will also often denote $x\cdot y$), the sign function $\sign{x}$
and that is closed under composition (when defined)  and linear length-ODE
scheme: The linear length-ODE
scheme
basically (a formal definition is provided in  Definition \ref{def:linear lengt ODE})   corresponds to defining functions from linear ODEs with respect to derivation along the length of the argument,
that is to say of the form $\dderivl{\tu f(x,\tu y)} = 	\tu A [\tu f(x,\tu y), 
x,\tu y]  \cdot \tu f(x,\tu y) 
+ \tu B [\tu f(x,\tu y), 
x,\tu y ].
$
Here, in the above description, we use the notation $\dderivl{\tu f(x,\tu y)}$, which corresponds to the derivation of $\tu f$ along the length function:  Given some function $\lengt:\N^{p+1} \rightarrow \Z$, and in particular for the case where $\lengt(x,\tu y)=\ell(x)$, 
\begin{equation}\label{lode}
\dderivL{\tu f(x,\tu y)}= \dderiv{\tu f(x,\tu y)}{\lengt(x,\tu
	y)} = \tu h(\tu f(x,\tu y),x,\tu y)
\end{equation}
is a formal synonym for
$ \tu f(x+1,\tu y)= \tu f(x,\tu y) + (\lengt(x+1,\tu y)-\lengt(x,\tu y)) \cdot
\tu h(\tu f(x,\tu y),x,\tu y).$

\begin{remark}
	This concept, introduced in  \THEPAPIERS, is motivated by the fact that the latter expression is similar to
	classical formula for classical continuous ODEs:
	$$\frac{\delta f(x,\tu y )}{\delta x} = \frac{\delta
		\lengt (x,\tu y) }{\delta x} \cdot \frac{\delta f(x,\tu
		y)}{\delta \lengt(x, \tu y)},$$
	and hence this is similar to a change of variable. Consequently, a linear length-ODE is basically a linear ODE over a variable $t$, once the change of variable $t=\ell(x)$ is done. 
\end{remark}

%
%
%

Call \emph{continuous function} a function of type $ f: S_{1} \times \dots \times S_{d} \to S'_{1} \times \dots S'_{d'}$ continuous in the sense of computable analysis, where each $S_{i},S'_{i}$ is either $\R$, $\N$ or $\Z$.
Considering that $\N \subset \R$, most of the basic functions and operations in this characterization (for example, $+$, $-$, \dots) have a clear meaning over the reals, i.e. are continuous functions.  As  ordinary differential equations are naturally living over the reals, it is rather natural to understand if one may go to computation theory for functions over the reals. We consider here computability and complexity over the reals in the most classical sense, that is to say, computable analysis (see e.g. \cite{Wei00}). So, a very natural question is to understand whether we can characterize $\FPtime$ for continuous functions, and not only discrete functions. This is one of the problem we solve in this article. 

\begin{remark} As in \cite{BlancBournezMCU22vantardise}, 
	clearly, we can consider $\N \subset \Z  \subset \R$, but as functions may have different types of outputs, composition is an issue. We simply admit that
	composition may not be defined in some cases. In other words, we consider that composition is a partial operator: for example, given $f: \N \to \R$ and $g: \R \to \R$, the composition of $g$ and $f$ is defined as expected, but $f$ cannot be composed with a function such as $h: \N \to \N$.
\end{remark}

Actually, there has been a first attempt in \cite{BlancBournezMCU22vantardise}, but the authors succeeded there only to characterize in an algebraic manner functions from the integers to the reals, while it would be more natural to talk about (or at least cover)  functions from the reals to the reals ($\|.\|$ stands for the sup-norm): Namely,  \cite{BlancBournezMCU22vantardise} considers  $\manonclass = [\mathbf{0},\mathbf{1},\projection{i}{k},   \length{x}, \plus, \minus, \times,\signb{x},\frac{x}{2};$ $composition$, $linear~length~ODE],$ where
\shortermcu{\begin{itemize}
		\item} $\ell: \N \to \N$ is the length function, mapping some integer to the length of its binary representation,  $\frac{x}{2}: \R \to \R$ is the function that divides by $2$, and all other basic functions are defined exactly as for $\linearderivlength$, but considered here as functions from the reals to reals. 
	%
	%
	%
	\shortermcu{
		\item}  $\signb{x}: \R \to \R$  is some piecewise affine function that 
	takes value $1$ for $x>\frac34$ and $0$ for $x<\frac14$, and is continuous piecewise affine: in particular, its restriction to the integers is the function $\sign{x}$ considered in $\linearderivlength$. 
	%
	%
	%
	\shortermcu{
	\end{itemize}
}
%
%
%

\begin{theorem}[{\cite{BlancBournezMCU22vantardise}}] \label{th:main:one}
	A function  $\tu f: \N^{d} \to \R^{d'}$ is computable in polynomial time if and only if there exists $\tilde{\tu f}:\N^{d+1} \to \R^{d'} \in \manonclass$ such that
	for all $\tu m \in \N^{d}$, $n \in \N$,
	$\|\tilde{\tu f}(\tu m,2^{n}) - \tu f(\tu m) \| \le 2^{-n}.$
\end{theorem}

A point is that their proof method uses some functions mapping in a continuous manner the integers into some suitable subsets of reals (namely the Cantor-like set $\Image$ corresponding to the reals whose radix 4 expansion is made of only $1$ and $3$). They observe this cannot be expanded to functions over the reals, as such a continuous mapping $\R \to \Image$ cannot exist, and leave open the question whether a similar result can be established for functions from the reals to the reals. 

In the current article, we solve the issue, by using an alternative proof method (but using some of the constructions from \cite{BlancBournezMCU22vantardise} that we extend in several directions):
%
First we observe that multiplication is not needed: 
Consider $$\manonclasslightidealsigmoid = [\mathbf{0},\mathbf{1},\projection{i}{k},   \length{x}, \plus, \minus,\signb{x},\frac{x}{2}, \frac{x}{3};{composition, linear~length~ODE}],$$ that is $\manonclass$ but without multiplication 
(when $I$ is some interval, we write $\tu x \in I$ when this holds componentwise). 

\begin{theorem}[Main theorem $1$, functions over the reals] \label{th:main:one:ex}
	A continuous function  $\tu f: \R^{d} \to \R^{d'}$ is computable in polynomial time if and only if there exists $\tilde{\tu f}:\R^{d} \times \N^{2}  \to \R^{d'} \in \manonclasslightidealsigmoid$ such that
	for all $\tu x \in \R^{d}$, $X \in \N$, $ \tu x \in\left[-2^{X}, 2^{X}\right]$, $n \in \N$,
	$\|\tilde{\tu f}(\tu x,2^{X}, 2^{n}) - \tu f(\tu x) \| \le 2^{-n}.$
\end{theorem}

This actually works for general continuous functions (Theorem \ref{th:main:one} is a special case):

\begin{theorem}[Main theorem $1$, general continuous functions] \label{th:main:one:ex:p}
	A function  $\tu f: \R^{d}  \times \N^{d''} \to \R^{d'}$ is computable in polynomial time iff there exists $\tilde{\tu f}:\R^{d} \times \N^{d''+2} \to \R^{d'} \in \manonclasslightidealsigmoid$ such that
	for all $\tu x \in \R^{d}$, $X \in \N$, $ \tu x \in\left[-2^{X}, 2^{X}\right]$, $\tu m  \in \N^{d''}$, $n \in \N$,
	$\|\tilde{\tu f}(\tu x, \tu m,2^{X},2^{n}) - \tu f(\tu x, \tu m) \| \le 2^{-n}.$
\end{theorem}

\manondufutur{
	Actually, we can even go further, and avoid the non-natural $\signb{x}$ function, and replace it by some smooth function such as hyperbolic tangent.
	
	$$\manonclasslighttanh = [\mathbf{0},\mathbf{1},\projection{i}{k},   \length{x}, \plus, \minus,\tanh,\frac{x}{2};\frac{x}{3};{composition, linear~length~ODE}],$$
	
	\begin{theorem}[Main theorem $2$, general continuous functions] \label{th:main:one:ex:pp}
		The statements of Theorem \ref{th:main:one:ex} and \ref{th:main:one:ex} remain true when replacing $\manonclasslightidealsigmoid$ by $\manonclasslighttanh$.
	\end{theorem}
}

This has also some strong links with formal neural networks:

\begin{theorem}[Main theorem $3$, Formal neural networks] \label{th:nn}
	Given some function $f: \R^{d} \to \R^{d'}$, for some given error and some polynomial time $t(n)$, we can efficiently produce some formal neural network that computes the function over its domain with the given precision, for any polynomial time computable $t(n)$  function $f$. 
\end{theorem}
%

Furthermore, this improves the characterization of \cite{MFCS2019}. Indeed, given a function $\tu f: \R^{d} \to \R^{d'}$, preserving the integers, we denote $\DP(f)$ for its discrete part: this is the function from $\N^{d} \to \N^{d'}$ whose value in $\vn \in \N^{d}$ is $\tu f(\tu n)$. Given a class $\mathcal{C}$ of such functions, we write $\DP(\mathcal{C})$ for the class of the discrete parts of the functions of $\mathcal{C}$. 
\begin{theorem} \label{trucchoseth}
	$\DP(\manonclasslightidealsigmoid)= \FPtime.$
\end{theorem}

%

\olivier{Parler de l'espace $PSPACE$?}

We can also extend the statements from \cite{BlancBournezMCU22vantardise}: 
%

\begin{definition}[Operation $\MANONlim$] Given $\tilde{\tu f}:\R^{d} \times \N^{d''} \times \N \to \R^{d'} \in \manonclasslightidealsigmoid$ such that
	for all $\tu x \in \R^{d}$,  $X \in \N$, $ \tu x \in\left[-2^{X}, 2^{X}\right]$, $\tu m \in \N^{d''}$, $n \in \N$,
	$\|\tilde{\tu f}(\tu x, \tu m,2^X,2^{n}) - \tu f(\tu x, \tu m) \| \le 2^{-n},$ then 
	$\MANONlim(\tilde{\tu f})$ is the (clearly uniquely defined) corresponding function  $\tu f: \R^{d} \to \R^{d'}$.
\end{definition}

%


\begin{theorem}  \label{th:main:two} 
	A continuous function $\tu f$ 
	is  computable in polynomial time if and only if all its components belong to $\manonclasslightidealsigmoidlim$, where
	$\manonclasslightidealsigmoidlim= [\mathbf{0},\mathbf{1},\projection{i}{k},   \length{x}, \plus, \minus,  \signb{x}, \frac{x}{2}, \frac{x}{3};$
	$composition$, $ linear~length~ODE;\MANONlim].$
\end{theorem}

{In particular:}

\begin{theorem}[Main theorem $4$]\label{th:main:twop} 
	$\manonclasslightidealsigmoidlim \cap \R^{\R} = \FPtime \cap \R^{\R} $
\end{theorem}

\shortermcu{
	The rest of the paper is organized as follows:}
\olivier{Vérifier découpage en sections dans ce texte}  In Section \ref{sec:discreteode}, we recall \shorter{some basic statements from }the theory of discrete ODEs. In Section \ref{sec:defanalysecalculable}, we recall  required concepts from computable analysis. In Section \ref{sec:functions}, we establish some properties about particular functions, required for our proofs. 
In
Section \ref{fptimedansmanonclass}, we prove that functions from $\manonclass$ are polynomial time computable, and then we 
prove a kind of reverse implication for functions over words. 
Section \ref{sec:computablereal} then proves Theorems \ref{th:main:two} and \ref{th:main:twop}. Section \ref{sec:generalizations} is a discussion about some of the consequences of our proofs. 
The appendix contains some complements, and missing proofs, and some complements on state of the art. Notice that some of the proofs in the main part of the documents are also repeated with more details in appendix. 

\subsection{Related work.}  \label{stateoftheart}

Recursion schemes constitute a classical major approach of classical computability theory and, to some extent, of complexity theory. The foundational characterization of $\FPtime$ due to Cobham \cite{cob65}, and then others based on safe recursion \cite{bs:impl} or
ramification (\cite{LM93,Lei94}), or for other classes \cite{lm:pspace}, gave birth to the very vivid field of \textit{implicit complexity} at the
interplay of logic and theory of programming: See 
\cite{Clo95,clote2013boolean} for monographs.

When considering continuous functions, various computability and complexity classes have been recently characterized using classical continuous ODEs: 
See survey  \cite{DBLP:journals/corr/BournezGP16}.

Here we are considering an approach in-between, where we characterize complexity classes using discrete ODEs. This approach is born from the  attempt of \THEPAPIERS{} to explain some of the constructions for continuous ODEs in an easier way. At the end, both models turn out to be rather different. Indeed, a key aspect of the proofs over the continuum is based on some closure properties, based on the formula for the derivative of a composition, while there is no such equivalent for discrete ODEs. However, an unexpected side effect of the approach was the characterizations obtained in \THEPAPIERS. 
They provided a  characterization of $\FPtime$ for discrete functions that does not make it necessary to specify an
explicit bound in the recursion, in contrast to Cobham's work \cite{cob65}, nor to assign
a specific role or type to variables, in contrast to safe recursion or ramification \cite{bs:impl,Lei-LCC94}. The characterization, like ours,  happens to be very simple 
using only natural notions from the world of ODE. In particular, considering \emph{linear} ordinary differential equations is something very natural in this context. 

Our proof is based on some constructions of the very recent \cite{BlancBournezMCU22vantardise}. However, we improve several of the statements and the constructions as we avoid multiplications. This latter paper was not able to characterize functions from the reals to the reals, but only sequences of reals.   Furthermore, our proof method, based on some adaptive barycenter is different, as their proof method  cannot extend to the reals for topological reasons. Our barycentric method is inspired from some constructions of \cite{BCGH07}, but once again the context of continuous ODEs and discrete ODEs is very different from the absence of a derivative formula for composition. This requires to construct explicitly some very particular functions as we do in Section \ref{sec:functions}. 

Our ways of simulating Turing machines have some reminiscence of similar constructions used in other contexts such as Neural Networks \cite{SS95,LivreSiegelmann}. But with respect to all previous contexts, as far as we know, only a few papers have been devoted to characterize complexity, and even computability, classes in the sense of computable analysis. There have been some attempts to 
the so-called $\R$-recursive functions \cite{DBLP:journals/corr/BournezGP16}. For discrete schemata, we only know \cite{brattka1996recursive}  and  \cite{ng2021recursion}, focusing on computability and not complexity.

\section{Some concepts from the theory of discrete ODEs}
\label{sec:discreteode}

In this section, we  recall some concepts and definitions from discrete ODEs, either well-known or established in \THEPAPIERSPLUS. Appendix \ref{sec:dode} presents the theory with many more details.
%
%
\newcommand\polynomial{ \fonction{sg}-polynomial}
%
%
%
%
%

\begin{definition}[{\cite{BlancBournezMCU22vantardise}}]
	A \polynomialb{}  expression $P(x_1,...,x_h)$ is an expression built-on
	$+,-,\times$ (often denoted $\cdot$) and $\signb{}$ functions over a set of variables $V=\{x_1,...,x_h\}$ and integer constants.
	The degree $\deg(x,P)$ of a term $x\in V$ in $P$ is defined inductively as follows:
	\shortermcu{
		\begin{itemize}
			\item} $\deg(x,x)=1$ and for  $x'\in V\cup \Z$ such that $x'\neq x$, $\deg(x,x')=0$;
		\shortermcu{\item}  $\deg(x,P+Q)=\max \{\deg(x,P),\deg(x,Q)\}$;
		\shortermcu{\item}   $\deg(x,P\times Q)=\deg(x,P)+\deg(x,Q)$;
		\shortermcu{\item}   $\deg(x,\signb{P})=0$.
		\shortermcu{
	\end{itemize}}
	A \polynomialb{}  expression $P$  is \textit{essentially constant} in
	$x$ if $\degre{x,P}=0$. 
\end{definition}

Compared to the classical notion of degree in polynomial expression,
all subterms that are within the scope of a sign (that is to say $\signb{}$)  function contributes
$0$ to the degree. A vectorial function (resp. a matrix or a vector) is said to be a \polynomialb{} expression if all
its coordinates (resp. coefficients) are. 
It is said to be
\textit{essentially constant} if all its coefficients are.

\begin{definition}[\THEPAPIERSPLUS] \label{def:essentiallylinear}
	A 
	\polynomialb{} expression $\tu g(\tu f(x, \tu y), x,
	\tu y)$ is \textit{essentially linear} in $\tu f(x, \tu y)$ if
	it is of the form $
	\tu g(\tu f(x, \tu y), x,
	\tu y) =
	\tu A [\tu f(x,\tu y), 
	x,\tu y]  \cdot \tu f(x,\tu y) 
	+ \tu B [\tu f(x,\tu y), 
	x,\tu y ] $
	where $\tu A$ and $\tu B$ are\polynomialb{} expressions essentially
	constant in $\tu f(x, \tu y)$.
\end{definition}

For example, 
the expression $P(x,y,z)=x\cdot \signb{(x^2-z)\cdot y} + y^3$
is essentially linear in $x$, essentially constant in $z$ and not linear in
$y$. 
\shortermcu{
	\item 
	The expression
	$P(x,2^{\length{y}},z)=\signb{x^2 - z}\cdot z^2 + 2^{\length{y}}$
	is essentially constant in $x$, essentially linear in
	$2^{\length{y}}$ (but not essentially constant) and not
	essentially linear in $z$. 
	\item }
The expression:
%
$  z +
(1-\signb{x})\cdot (1-\signb{-x})\cdot (y-z) $
is essentially constant in $x$ and linear in $y$ and $z$.
%

\begin{definition}[Linear length ODE \THEPAPIERS]\label{def:linear lengt ODE}
	A function $\tu f$ is linear $\lengt$-ODE definable (from $\tu u$, 
	$\tu g$ and $\tu h$) if it corresponds to the
	solution of 
	%
	\begin{equation} \label{SPLode}
	f(0,\tu y) 
	= \tu g(\tu y)   \quad  and \quad
	\dderivl{\tu f(x,\tu y)} 
	=   \tu u(\tu f(x,\tu y), \tu h(x,\tu y),
	x,\tu y) 
	\end{equation}
	\noindent where $\tu u$ is \textit{essentially linear} in $\tu f(x, \tu y)$. 
\end{definition}

A fundamental fact is that the derivation with respects to length provides a way to do a kind of change of variables: consequently, we will often define some functions by defining their value in 
$2^{0}$, and then $2^{n+1}$ from their value in $2^{n}$ as its corresponds to some discrete ODE after this change of variable. We will also implicitly use that some basic functions such as $n \mapsto 2^{n}$ can easily be defined, and that we can produce $2^{T(\ell(\omega))}$ for any polynomial $T$: see \THEPAPIERS.

\begin{lemma}[{\THEPAPIERS}]
	Let 
	$f: \N^{p+1}\rightarrow \Z^d$,
	$\lengt:\N^{p+1}\rightarrow \Z$  be some functions and assume that \eqref{lode} holds considering  $\lengt(x,\tu y)=\length{x}$.
	Then $\tu f(x,\tu y)$ is given by 
	$\tu f(x,\tu y)= \tu F(\length{x},\tu y)$
	where $\tu F$ is the solution of the initial value problem $F(1,\tu y)= \tu f(0,\tu y)$, and $\dderiv{\tu F(t,\tu y)}{t} =\tu h(\tu F(t, \tu y),2^{t}-1,\tu y).$
\end{lemma}



\section{Concepts from computable analysis}
\label{sec:defanalysecalculable}


When we say that a function $f: S_{1}  \times \dots \times S_{d} \to \R^{d'}$ is (respectively: polynomial-time) computable this will always be in the sense of computable analysis: see e.g. \cite{brattka2008tutorial,Wei00}. We recall here the basic concepts and definitions, mostly following the book \cite{Ko91}, whose subject is complexity theory in computable analysis. This section is basically repeating the formalization proposed in \cite{BlancBournezMCU22vantardise} done to mix complexity issues dealing with integer and real arguments: 
%
%
%
%
%
%
%
%
a dyadic  number $d$ is a rational number with  a finite binary expansion. That is to say $d=m / 2^{n}$ for some integers $m \in \Z$, $n\in \N$, $n \geq 0$. Let $\dyadic$ be the set of all dyadic rational numbers. We denote by $\dyadic_{n}$ the set of all dyadic rationals $d$ with a representation $s$ of precision $\operatorname{prec}(s)=n$; that is, $\dyadic_{n}=\left\{m \cdot 2^{-n} \mid m \in \Z\right\}$.

\begin{definition}[\cite{Ko91}]  \label{def:cinq} For each real number $x$, a function $\phi: \N \rightarrow \dyadic$ is said to binary converge to $x$ if  for all $n \in \N, \operatorname{prec}(\phi(n))=n$ and $|\phi(n)-x| \leq 2^{-n}$. Let $C F_{x}$ (Cauchy function) denotes the set of all functions binary converging to $x$.
\end{definition}

%
%
%
%
%

Intuitively, a Turing machine $M$ computes a real function $f$ the following way: 1. The input $x$ to $f$, represented by some $\phi \in C F_{x}$, is given to $M$ as an oracle; 2. The output precision $2^{-n}$ is given in the form of integer $n$  as the input to $M$; 3. The computation of $M$ usually takes two steps, though sometimes these two steps may be repeated for an indefinite number of times;
4. $M$ computes, from the output precision $2^{-n}$, the required input precision $2^{-m}$; 5. $M$ queries the oracle to get $\phi(m)$, such that $\|\phi(m)-x\| \leq 2^{-m}$, and computes from $\phi(m)$ an output $d \in \dyadic$ with $\|d-f(x)\| \leq$ $2^{-n}$.
%
More formally:

\begin{definition}[\cite{Ko91}] A real function $f: \R \rightarrow \R$ is computable if there is a function-oracle {TM} $M$ such that for each $x \in \R$ and each $\phi \in C F_{x}$, the function $\psi$ computed by $M$ with oracle $\phi$ (i.e., $\left.\psi(n)=M^{\phi}(n)\right)$ is in $C F_{f(x)}$. 
	\shortermcu{We say the function $f$ is computable on interval $[a, b]$ if the above condition holds for all $x \in[a, b]$.}
\end{definition}

\shortermcu{
	\begin{remark}
		Given some $x \in \R$, such  an oracle TM $M$ can determine some integer $X$ such that $x \in [-2^{X},2^{X}]$.
	\end{remark}
}

\olivier{Reste de MCU: Probablement superflus pour l'instant:
	The following concept plays a very important role:

	\begin{definition} \label{def:above}
		Let $f:[a, b] \rightarrow \R$ be a continuous function on $[a, b]$. Then, a function $m: \N \rightarrow \N$ is said to be a modulus function of $f$ on $[a, b]$ if for all $n \in \N$ and all $x, y \in[a, b]$, we have
		$$
		|x-y| \leq 2^{-m(n)} \Rightarrow|f(x)-f(y)| \leq 2^{-n}
		$$
	\end{definition}
	
	The following is well known (see e.g. \cite{Ko91} for a proof):

	\begin{theorem}
		A function $f: \R \rightarrow \R$ is computable iff there exist two recursive functions $m: \N \times \N \rightarrow \N$ and $\psi: \dyadic \times \N \rightarrow \dyadic$ such that
		\begin{enumerate}
			\item for all $k, n \in \N$ and all $x, y \in[-k, k],|x-y| \leq 2^{-m(k, n)}$ implies $|f(x)-f(y)| \leq 2^{-n}$, and
			\item  for all $d \in \dyadic$ and all $n \in \N,|\psi(d, n)-f(d)| \leq 2^{-n}$.
		\end{enumerate}
	\end{theorem}
}

\shortermcu{
	\subsection{On computable analysis: Complexity}
}

%
%
%
Assume that $M$ is an oracle machine which computes $f$ on a do$\operatorname{main} G$. For any oracle $\phi \in C F_{x}$, with $x \in G$, let $T_{M}(\phi, n)$ be the number of steps for $M$ to halt on input $n$ with oracle $\phi$, and $T_{M}^{\prime}(x, n)=\max \left\{T_{M}(\phi, n) \mid \phi \in C F_{x}\right\}$. The time complexity of $f$ is defined as follows:

\begin{definition}[\cite{Ko91}]
	Let $G$ be  bounded closed interval $[a, b]$. Let $f: G \rightarrow \R$ be a computable function. Then, we say that the time complexity of $f$ on $G$ is bounded by a function $t: G \times \N \rightarrow \N$ if there exists an oracle TM $M$ which computes $f$ 
	such that for all $x \in G$ and all $n>0$, $T_{M}^{\prime}(x, n) \leq t(x, n)$.
\end{definition}

In other words, the idea is to measure the time complexity of a real function based on two parameters: input real number $x$ and output precision $2^{-n}$. Sometimes, it  is more convenient to simplify the complexity measure to be based on only one parameter, the output precision. For this purpose, we say the uniform time complexity of $f$ on $G$ is bounded by a function $t^{\prime}: \N \rightarrow \N$ if the time complexity of $f$ on $G$ is bounded by a function $t: G \times \N \rightarrow \N$ with the property that for all $x \in G$, $t(x, n) \leq t^{\prime}(n)$.

{
	However, if we do so, it is important to realize that if we had taken $G=\R$ in the previous definition, for unbounded functions $f$, the uniform time complexity would not have existed, because the number of moves required to write down the integral part of $f(x)$ grows as $x$ approaches $+\infty$ or $-\infty$. Therefore, the approach of \cite{Ko91} is to do as follows (The bounds $-2^{X}$ and $2^{X}$ are somewhat arbitrary, but are  chosen here  because the binary expansion of any $x \in\left(-2^{n}, 2^{n}\right)$ has at most $n$ bits in the integral part).
	%
	%
	\begin{definition}[Adapted from \cite{Ko91}]  For functions $f(x)$ whose domain is $\R$, 
		we say that the (non-uniform) time complexity of $f$ is bounded by a function $t^{\prime}: \N^{2} \rightarrow \N$ if the time complexity of $f$ on $\left[-2^{X}, 2^{X}\right]$ is bounded by a function $t: \N^{2} \rightarrow \N$ such that $t(x, n) \leq t^{\prime}(X, n)$ for all $x \in\left[-2^{X}, 2^{X}\right]$. 
	\end{definition}

	\olivier{RESTE DE MCU. probablement superflus car on parle pas d'espace. Mais prédisement, il FAUT EN PARLER !
		The space complexity of a real function is defined in a similar way. We say the space complexity of $f: G \rightarrow \R$ is bounded by a function $s: G \times \N \rightarrow \N$ if there is an oracle TM $M$ which computes $f$ such that for any input $n$ and any oracle $\phi \in C F_{x}, M^{\phi}(n)$ uses $\leq s(x, n)$ cells, and the uniform space complexity of $f$ is bounded by $s^{\prime}: \N \rightarrow \N$ if for all $x \in G$ and all $\phi \in C F_{x}, M^{\phi}(n)$ uses $\leq s^{\prime}(n)$ cells.
	}
	
	%
	
	As we want to talk about general functions in $\Lesfonctionsquoi$, we extend the approach to more general functions.  
	%
	%
	(for conciseness, when $\tu x=(x_{1},\dots,x_{p})$, $\tu X= (X_{1},\dots, X_{p})$, we write
	$\tu x \in [-2^{\tu X}, 2^{\tu X}]$ as a shortcut for $x_{1} \in\left[-2^{X_{1}}, 2^{X_{1}}\right]$,  \dots, $x_{p} \in\left[-2^{X_{p}}, 2^{X_{p}}\right]$). 
	
	
	\begin{definition}[Complexity for real functions: general case]   \label{def:bonendroit} Consider a function  $f(x_{1},\dots$ $,x_{p}$, $n_{1},\dots,n_{q})$ whose domain is $\R^{p} \times \N^{q}$. 
		We say that the (non-uniform) time complexity of $f$ is bounded by a function $t^{\prime}: \N^{p+q+1} \rightarrow \N$ if the time complexity of $f(\cdot,\dots,\cdot,\ell(n_{1}),\dots,\ell(n_{q}))$ on $\left[-2^{X_{1}}, 2^{X_{1}}\right] \times \dots \left[-2^{X_{p}}, 2^{X_{p}}\right] $  
		is bounded by a function $t(\cdot,\dots,\cdot,\ell(n_{1}),\dots,\ell(n_{q}),\cdot): \N^{p} \times \N \to \N$ such that 
		$$ t(\tu x,\ell(n_{1}),\dots,\ell(n_{q}), n) \leq t^{\prime}(\tu X,\ell(n_{1}), \dots,\ell(n_{q}), n)$$
		whenever  $\tu x \in \left[-2^{\tu X}, 2^{\tu X}\right].$
		We say that $f$ is polynomial time computable if $t^{\prime}$ can be chosen as a polynomial. 
		We say that a vectorial function is polynomial time computable iff all its components are. 
	\end{definition}
	
	\shortermcu{
		\begin{remark}
			There is some important {subtlety}: When considering $f: \N \to \Q$, as $\Q \subset \R$, stating $f$ is computable may mean two things: in the classical sense, given integer $y$,  i.e. one can compute $p_y$ and $q_{y}$ some integers such that $f(y)=p_{y}/q_{y}$, or that it is computable in the sense of computable analysis: given some precision $n$,  given arbitrary $y$, and $n$ we can provide some rational (or even dyadic) $q_{n}$ such that $|q_{n}-f(y)| \leq 2^{-n}$. As we said, we always consider the latter.
		\end{remark}
	}
	
	We do that so this measures of complexity extends the usual complexity measure for functions over the integers, where complexity of integers is measured with respects of their lengths, and over the reals, where complexity is measured with respect to their approximation.
	%
	In particular, in the specific case of a function $f: \N^{d} \to \R^{d'}$, that basically means there is some polynomial $t': \N^{d+1} \to \N$ so that the time complexity of producing some dyadic approximating $f(\tu m)$ at precision $2^{-n}$ is bounded by $t'(\ell(m_{1}),\dots,\ell(m_{d}),n)$. 
	
	
	%
	
	In other words, when considering that a function is polynomial time computable, it is in the length of all its integer arguments, as this is the usual convention. However, we need sometimes to consider also polynomial dependency directly
	in one of some specific integer argument, say $n_{i}$,  and not on its length $\ell(n_{i})$. We say that \emph{the  function is polynomial time computable, \unaire{n_{i}}} when this holds (keeping possible other integer arguments $n_{j}$, $j \neq i$, measured by their length). 
	
	\olivier{Commenté, car pas besoin, je pense (enfin j'espère)
		A well-known observation is the following.
		
		\begin{theorem} Consider $\tu f$ as in Definition \eqref{def:bonendroit} computable in polynomial time. Then $\tu f$ has a polynomial modulus function of continuity, that is to say there is a polynomial function $m_{\tu f}: \N^{p+q+1}\rightarrow \N$ such that for all $\tu x,\tu y$ and all $n>0$, $\|\tu x-\tu y\| \leq 2^{-m_{\tu f}(\tu X,\ell(n_{1}),\dots,\ell(n_{q}), n)}$ implies 
			$\|\tu f(\tu x,n_{1}, \dots,n_{q})-\tu f(\tu y,n_{1}, \dots,n_{q})\| \leq 2^{-n}$,
			whenever $\tu x,\tu y  \in\left[-2^{\tu X}, 2^{\tu X}\right]$.
		\end{theorem}
	}

	\section{Some results about various functions}
	\label{sec:functions}

	\olivier{Besoin d'une def. Est-ce clair?}
	Call \motnouv{affine function} a function $f: \R^{n} \to \R$ of the form $f(x_{1},\dots,x_{n})= w_{1} x_{1} + \dots + w_{n} x_{n} +h$, for some real $w_{1},\dots, w_{n}$ or a function $f: \R^{n} \to \R^{m}$ whose $m$ components are of this form. We call \motnouv{neural function} the smallest class of functions that is obtained by considering an affine function, or an affine function where one of several  of its variable has  been replaced by $\signb{g}$ where $g$ is inductively a neural function.
	\begin{remark}
		The idea is to capture the class of functions computed by some (non-recurrent) formal neural network: 
		every neural function can clearly be interpreted as a formal neural network where the activation function is the function $\signb{}$. 
		As an example, $x+ \signb{x+2 y +3 \signb{x}}$ is a neural function, which can be considered as a depth $2$ formal neural network. A function such as $x^{2}+2$ is not a neural function, as it involves a multiplication.
	\end{remark}
	\manondufutur{
		A \motnouv{neural function} is a neural function with the $\tanh$ sigmoid.
	}
	
	\manondufutur{
		Ai pensé à un moment pour éviter division par 3.
		
		Answer:  Right in the current writing, we need division by 3. This comes from the 3/4 sig(1/8,7/8,x) line 746, as pointed out. A solution would be to put Division by 3 among the basis functions in the definition l 140.
		
		With this corrected, our main statements and claim are valid: multiplication can be avoided, and we do have a characterization of polynomial time, without changing the text.
		
		But actually, there is no need to do so. Lemma 22 was stated in order to provide an easy proof of Lemma 23  (with most details in the appendix).  If division by 3 is not authorized as in definition l 140, it can indeed be avoided by using density of the dyadics: if 3/4 sig(1/8,7/8,x) is replaced by beta' sig(1/8,1/8 + 1/beta,x), (with suitable dyadics beta,beta', namely beta'=1, beta=5/4 are fine) then we would not compute {x-1/8} but alpha {x-1/8} for dyadic alpha=beta/beta'. Then, fractional parts should be multipled by this factor alpha in the statement of Lemma 22. Now, replacing σi(2n, x) = x−ξi(2n, x) by σi(2n, x) = alpha x−ξi(2n, x) would still lead to Lemma 23, as verified with Maple (correct statements have also been updated on ArXiv meanwhile for Lemma 22).

		MAIS AURAIT partieentier de \alphax x et pas de x.
		
		Maple montre que ok pour alpha=5/4. Mais pondéfration ensuite doit faire disparaitre le 5.
		
		And then all other statements and proofs would remain valid.
		
	}
	
	
	A key part of our proofs is the construction of very specific functions in $\manonclasslightidealsigmoid$: we write $\{x\}$ for the fractional part of real $x$, i.e. $\{x\}=x-\lfloor x \rfloor$. We provide more details and show some graphical representations of most of them in the appendix, in order in particular to show that these functions are sometimes highly non-trivial.
	\begin{lemma}\label{lem:xi}
		There exists $\xi_1, \xi_2 : \N \times \R \mapsto \R ~  \in \manonclasslightidealsigmoid$  
		such that, for all $n\in \N$ and $x\in [- 2^{n} , 2^{n}] $, 
		\shortermcu{\begin{itemize}
				\item} whenever $ x \in [\lfloor x \rfloor - \frac{1}{2}, \lfloor x \rfloor + \frac{1}{4}] $  , $\xi_1(2^n,x) = \{ x \} $, 
			\shortermcu{\item} and whenever $ x \in [\lfloor x \rfloor, \lfloor x \rfloor + \frac{3}{4}] $  , $\xi_2(2^n,x) = \{ x \} $.
			\shortermcu{\end{itemize}}
	\end{lemma}
	
	\begin{proof}
		Consider 
		$\xi_{1}(N,x)=\xi(N,x-\frac58) -\frac14$ 
		and $\xi_{2}(N,x)=\xi(N,x+\frac18)$, where $\xi(N,x)$ $=\xi'(N+1,x)-\xi'(N+1,-x)+\frac34-\frac34 \signb{\frac14+4x}$ 
		where $\xi'(2^{0},x) = \frac34 \signb{\frac16+\frac23 x}$, 
		and $\xi'(2^{n+1},x)= \xi'(2^{n},F(2^{n},x))$, with
		$F(K,x)= x- K.\signb{\frac14+4(x-K)}$, 
		considering $F(0,x)=\xi'(2^{0},x)$.
		A proof by induction (more intuitions and details in appendix) shows that it satisfies our claims. It remains to prove that 
		this corresponds to a function in $\manonclasslightidealsigmoid$, but the key is to observe that, from an easy induction,  $\xi'(2^{n},x)= F(2^{0},F(2^{1},F(2^{2}(\dots, F(2^{n-1},x)))))$, and hence
		can be obtained as $H(2^{n-1},2^{n},x)$ with $H$ defined by some linear length ordinary differential equation using the derivative on its first variable expressing the recurrence
		$H(2^{0}, 2^n, x) = F(2^{n-1}, x)$ and $H(2^{t+1}, 2^n, x) = F(2^{n-1-t}, H(2^{t}, 2^n, x))$.
	\end{proof}
	
	Considering $\sigma_i (2^n,x) = x - \xi_i(2^n,x) $, we obtain next lemma. Using  recursive constructions, we can also get (details and graphical representations in appendix).

	\begin{lemma}\label{lem:i}
		There exists $\sigma_1, \sigma_2 : \N \times \R \mapsto \R ~  \in \manonclasslightidealsigmoid$  
		such that, for all $n\in \N$ and $x\in [-2^{n} , 2^{n}]$, 
		\shortermcu{\begin{itemize}
				\item} whenever  $ x \in I_{1}=[\lfloor x \rfloor - \frac{1}{2}, \lfloor x \rfloor + \frac{1}{4}] $  , $\sigma_1(2^n,x) = \lfloor x \rfloor $, 
			\shortermcu{\item} and whenever  $ x \in I_{2}=[\lfloor x \rfloor, \lfloor x \rfloor + \frac{3}{4}] $  , $\sigma_2(2^n,x) = \lfloor x \rfloor $.
			\shortermcu{\end{itemize}}
	\end{lemma}

	\begin{lemma}\label{lem:mod2}
		There exists $\mod_{2} : \N\times\R \mapsto [0,1]   \in \manonclasslightidealsigmoid$  
		such that for all $n\in\N$, $x\in [-2^{n} , 2^{n}]$, 
		\shortermcu{\begin{itemize}
				\item} whenever  $ x \in [\lfloor x \rfloor - \frac{1}{4}, \lfloor x \rfloor + \frac{1}{2}] $, $\mod_{2}(x)$ is $\lfloor x \rfloor$ modulo $2$.
			\shortermcu{\end{itemize}}
	\end{lemma}

	\begin{lemma}\label{lem:div2}
		There exists $\div_{2} : \N\times\R \mapsto [0,1]  \in \manonclasslightidealsigmoid$  
		such that for all $n\in\N$, $x\in [-2^{n} , 2^{n}]$, 
		\shortermcu{\begin{itemize}
				\item} whenever  $ x \in [\lfloor x \rfloor, \lfloor x \rfloor + \frac{1}{2}] $, $\div_{2}(x)$ is the integer division of $\lfloor x \rfloor$ by $2$.
			\shortermcu{\end{itemize}}
	\end{lemma}
	
	
	\begin{lemma}\label{lem:lambda}
		There exists $\lambda : \N\times\R \mapsto [0,1]   \in \manonclasslightidealsigmoid$  
		such that for all $n\in\N$, $x\in [-2^{n} , 2^{n}]$, 
		\shortermcu{\begin{itemize}
				\item} whenever  $ x \in [\lfloor x \rfloor + \frac{1}{4}, \lfloor x \rfloor + \frac{1}{2}] $  , $\lambda(2^n,x) = 0$ 
			\shortermcu{\item} and whenever $  x \in [\lfloor x \rfloor + \frac{3}{4}, \lfloor x \rfloor +1] $, $\lambda(2^n,x) = 1$.
			\shortermcu{\end{itemize}}
		\olivier{Mis ailleurs ce texte:
			In particular, whenever we know from the above constraints 
			\begin{itemize}
				\item that $\lambda(2^n,x) = 0$, we know that $\sigma_{2}(2^n,x) = \lfloor x \rfloor$, 
				\item or that $\lambda(2^n,x) = 1$ we know that $\sigma_{1}(2^n,x) = \lfloor x \rfloor$, 
				\item or that  $\lambda(2^n,x) \in (0,1)$, we know that $\sigma_1(2^n,x) = \lfloor x \rfloor +1 $ and $\sigma_2(2^n,x) = \lfloor x \rfloor$.
		\end{itemize}}
	\end{lemma}
	
	\manon{
		=======
		\begin{proof}
			TODO $\sigma $ sur le bon intervalle
			
			TODO 
		\end{proof}
	}
	
	\shortermcu{
		\subsection{Some particular functions}
	}
	
	
	\shortermcu{
		\subsection{Some error bounds}
	}
	
	\manondufutur{
		\olivier{Je vole ici des trucs d'autres sources: possible que cette section doit etre avant l'autre en fait. Par ailleurs, probablement beaucoup de résultats de cette section qu'on utilisera jamais}
	}
	
	\manondufutur{
		When simulating a Turing machine, we will also need to keep the error
		under control. In many cases, this will be done with the help of the
		error-contracting function $\danielsigmanameold: \R \to \R$ defined by
		
		$$\danielsigmanameold: x \mapsto \danielsigmaold{x} = x - 0.2 \sin (2 \pi x).$$

		The function $\danielsigmanameold$ is a contraction on the vicinity of integers:

		\begin{lemma}[{\cite[Lemma 4.2.1]{TheseDaniel}}]\label{gcb1}
			Let $n \in \Z$, and let $\epsilon \in [0,1/2)$. Then there is some
			contracting factor $\lambda_\epsilon \in (0,1)$ such that $\forall
			\delta \in [-\epsilon,\epsilon]$, $\left|\danielsigmaold{n+\delta} - n\right| <
			\lambda_\epsilon \delta$.
		\end{lemma}
		
		\begin{proof}
			It is sufficient to consider the case where $n=0 .$ Because $\danielsigmanameold$ is odd, we only study $\danielsigmaold$ in the interval $[0, \varepsilon] .$ Let $g(x)=\danielsigmaold{x} / x .$ This function is strictly increasing in $(0,1 / 2] .$ Then, noting that $g(1 / 2)=1$ and $\lim _{x \rightarrow 0} g(x)=1-0.4 \pi \approx-0.256637,$ we conclude that there exists some $\lambda_{\varepsilon} \in(0,1)$ such that $\left|\danielsigmaold{x}\right|<\lambda_{\varepsilon}|x|$ for all $x \in[-\varepsilon, \varepsilon]$
		\end{proof}
	}
	\manondufutur{
		\olivier{Euh, copie colle, mais Vrai ce paragraphe, qu'on fait ces conventions
			
			For the rest of this document, we suppose that $\epsilon \in
			[0, 1/2)$ is fixed and that $\lambda_\epsilon$  is the respective
			contracting factor given by Lemma \ref{gcb1}. The function $\danielsigmanameold$ will be used in our simulation to keep the error
			controlled when bounded quantities are involved (e.g., the actual
			state, the symbol being read, etc.).
	}}
	\manondufutur{
		To help readability, we will also write $\danielsigma{x}$ for $\danielsigmaold{x}$. We will also write $\danielsigmae{[n]}{x}$ for $\danielsigmaeold{[n]}{x}$.
		The idea is that this is intended to be close to $x$, when $x$ is some integer.
	}
	\manondufutur{
		\olivier{Euh, hors sujet, non?
			We will also need another error-contracting function that controls the
			error for unbounded quantities. This will be achieved with the help of
			the function $\ltroisname : \R^2 \to  \R$, that has the property that whenever
			$\overline{a}$ is an
			approximation of $a \in  \{0,1,2\}$, then $\left|\ltrois{\overline{a}}{y} -a \right| <
			1/y$, for $y>
			0$.
			In
			other words, $\ltrois{.}{y}$  is an error- contracting map, where the error is
			contracted by an amount specified by the  argument $y$.
	}}
	\manondufutur{
		\begin{lemma}[{\cite[Lemma 4.2.3]{TheseDaniel}}]
			\label{Lemma 4.2.3} $\left|\frac{\pi}{2}-\arctan x\right|<\frac{1}{x}$ for $x \in(0, \infty)$
		\end{lemma}
		
		\begin{proof}
			Let $f(x)=\frac{1}{x}+\arctan x-\frac{\pi}{2} .$ It is easy to see that $f$ is decreasing in $(0, \infty)$ and that $\lim _{x \rightarrow \infty} f(x)=0 .$ Therefore $f(x)>0$ for $x \in(0, \infty)$ and the result holds.
		\end{proof}
	}
	\manondufutur{
		\begin{lemma}\label{maplenew}\label{MonLemma 4.2.3}
			$\left|1-\tanh x\right|<2 \exp(-2x)$ for $x \in(0, \infty)$
		\end{lemma}
		
		\begin{proof}
			Let $f(x)=2 \exp(-2x) + \tanh(x) -1$. Then its derivative is $1 - \tanh(x)^2 - 4\exp(-2x)= 1- \frac{1 - exp(-2x)}{1 + \exp(-2x)}^2 - 4\exp(-2 x)=-4\exp(-4x)\frac{2+\exp(-2x) }{(1 + exp(-2x))^2}$
			which is negative for $x \in(0, \infty)$, and hence $f$ is decreasing on $(0, \infty)$. As its limits in $+\infty$ is $0$, we have $f(x)>0$ on $(0, \infty)$, and the result holds.
		\end{proof}
	}
	\manondufutur{
		New (merci maple)
		\begin{lemma}[{\cite[Lemma 4.2.4]{TheseDaniel}}]
			\label{Lemma 4.2.4} $\left|\frac{\pi}{2}+\arctan x\right|<\frac{1}{|x|}$ for $x \in(-\infty, 0)$.
		\end{lemma}
		
		\begin{proof}
			Take $f(x)=\frac{1}{x}+\arctan x+\frac{\pi}{2}$ and proceed as in Lemma \ref{Lemma 4.2.3}.
		\end{proof}
	}
	\manondufutur{
		\begin{lemma}\label{MonLemma 4.2.4}
			$\left|1+\tanh x\right|<2 \exp(-2|x|) $ for  $x \in(-\infty, 0)$.
		\end{lemma}
		
		\begin{proof}
			Let $f(x)=2 \exp(2x) - \tanh(x) -1$. We have $f(-x)= 2 \exp(-2x) + \tanh(x) -1$. From the proof of Lemma \ref{maplenew}, studying the same function, we know that $f(-x)>0$ on $(0, \infty)$. Hence, $f(x)>0$ on $(-\infty, 0)$, and the result holds.
		\end{proof}
	}
	\manondufutur{
		In order to define function $\ltroisname$ we first define a preliminary function $\ldeuxname$ satisfying similar
		conditions, but only when $a \in  \{0, 1\}$.
		
		\begin{lemma}[{\cite[Lemma 4.2.5]{TheseDaniel}}] \label{dsg:lemma2}
			Let $\ldeuxname: \R^2 \to \R$ be given by $$\ldeux{x}{y} = \frac1\pi
			\arctan(4y(x-1/2))+\frac12.$$ Suppose that $a \in \{0,1\}$. Then, for any
			$\overline{a},y \in \R$ satisfying $|a-\overline{a}| \le 1/4$ and
			$y>0$, we get $\left|\ldeux{\overline{a}}{y} - a \right| < 1/y.$
		\end{lemma}

		\begin{proof}
			\begin{itemize}
				\item  Consider $a=0 .$ Then $\bar{a}-1 / 2 \leq-1 / 4$ implies $|4 y(\bar{a}-1 / 2)| \geq y .$ Therefore, by Lemma \ref{Lemma 4.2.4}
				$$
				\left|\frac{\pi}{2}+\arctan (4 y(\bar{a}-1 / 2))\right|<\frac{1}{|4 y(\bar{a}-1 / 2)|} \leq \frac{1}{y}.
				$$
				Moreover, multiplying the last inequality by $1 / \pi$ and noting that $\frac{1}{\pi y}<\frac{1}{y},$ it follows that
				$$
				\left|\ldeux{\bar{a}}{ y} -a \right|<1 / y.
				$$
				\item  Consider $a=1$. Remark that $\bar{a}-1 / 2 \geq 1 / 4$ and proceed as above, using Lemma \ref{Lemma 4.2.3} instead of Lemma \ref{Lemma 4.2.4}.
			\end{itemize}
		\end{proof}
	}
	\manondufutur{
		\olivier{On peut faire une fonction avec la même propriété avec $\tanh$ . Est-ce utile?}
		
		\begin{lemma}
			Let $\ldeuxname: \R^2 \to \R$ be given by $$\ldeux{x}{y} = 
			\frac12 \tanh(4y(x-1/2)) + \frac12 $$ Suppose that $a \in \{0,1\}$. Then, for any
			$\overline{a},y \in \R$ satisfying $|a-\overline{a}| \le 1/4$ and
			$y>0$, we get $\left|\ldeux{\overline{a}}{y} - a \right| <  \exp(-2y).$
		\end{lemma}
	}
	\olivier{Vérifier ce délire/calcul.}
	\manondufutur{
		\begin{proof}
			\begin{itemize}
				\item  Consider $a=0 .$ Then $\bar{a}-1 / 2 \leq-1 / 4$ implies $|4 y(\bar{a}-1 / 2)| \geq y .$ Therefore, by Lemma \ref{MonLemma 4.2.3}
				$$
				\left|1+\tanh (4 y(\bar{a}-1 / 2))\right|<\frac{2}{\exp(2|4 y(\bar{a}-1 / 2|)}\leq \frac{2}{\exp(2y)}.
				$$
				
				Multiplying the last inequality by $1 / 2$, it follows that
				$$
				\left|\ldeux{\bar{a}}{ y} -a \right|<\exp(-2y).
				$$

				\item  Consider $a=1$. Remark that $\bar{a}-1 / 2 \geq 1 / 4$ and proceed as above, using Lemma \ref{MonLemma 4.2.3} instead of Lemma \ref{MonLemma 4.2.4}.
			\end{itemize}
			
		\end{proof}
	}
	\olivier{New: a new bound on $\tanh(x)$}
	\manondufutur{
		\begin{lemma}
			For all $x\ge 0$, we have $$x-\frac{x^{3}}{3} \le \tanh(x) \le x$$.
		\end{lemma}
		
		\begin{proof}
			Consider $f(x)=x-\tanh(x)$. Its derivative is $\tanh^{2}(x)$, that is non-negative for $x \ge 0$. Consequently, it is increasing, and as it values $0$ in $0$, it is always non-negative. This proves that $\tanh(x) \le x$ for all $x \ge 0$.
			
			Consider $g(x)= \tanh(x) -x +  \frac{x^{3}}{3}$. Its derivative is $x^{2} - \tanh(x)^{2}= (x+\tanh(x)) (x- \tanh(x)) = (x+\tanh(x)) f(x)$ that is non-negative for $x \ge 0$ from above. Consequently, it is increasing, and as it values $0$ in $0$, it is always non-negative. This proves that $x-\frac{x^{3}}{3} \le \tanh(x)$  for all $x \ge 0$.
			
		\end{proof}
	}
	\manondufutur{
		\begin{lemma}
			For all $x \le 0$, we have $$x \le \tanh(x) \le x-\frac{x^{3}}{3}$$.
		\end{lemma}
		
		\begin{proof}
			Then $-x \ge 0$, and hence $-x+\frac{x^{3}}{3} \le \tanh(-x) \le -x$ from previous lemma, from which the result follows.
		\end{proof}
	}
	\manondufutur{
		\begin{corollary}
			For all $x$, $$| \tanh(x) - x| \le \frac{|x|^{3}}{3}$$
		\end{corollary}
	}
	\manondufutur{
		\begin{corollary}
			For all $K>0$, 
			For all $x$, $$| K \tanh(x/K) - x | \le \frac{|x|^{3}}{3K^{2}}$$
		\end{corollary}
		
		\begin{proof}
			Apply previous corollary on $x/K$, to get $| \tanh(x/K) - x/K| \le \frac{|x|^{3}}{K^{3}3}$, and then multiply by $K$. 
		\end{proof}
	}
	
	\shortermcu{
		\subsection{Some properties of sigmoids}
	}
	
	%
	%
	\begin{lemma} \label{tricksigmoid}
		Consider $T(d,l) = \signb{d-3/4+l/2}$.
		For $l \in [0,1]$, we have $T(0,l)=0$, and $T(1,l)=l$.
	\end{lemma}
	
	\manondufutur{
		If we want to do something similar using $\tanh$, from above descriptions, $$T(d,l)=100 \tanh(l/100 + (d - 1) 100) + 100 \tanh(-d + 1)/\tanh(1)$$
		is very close to $0$ for $d=0$, and very close to $l$ for $d=1$, when $l \in [-1..1]$.  But, NOTICE THAT THIS expression requires the inverse of $\tanh(1)$.
	}
	
	\manondufutur{
		\begin{lemma}\label{lem:sigmoids}
			An expression such as $\If(q-\alpha_{n},T_{n},\If(q-\alpha_{n-1},T_{n-1},\dots,\If(q-\alpha_{1},T_{1},E_{1})))$, for integers $\alpha_{1} < \alpha_{2} <\dots, \alpha_{n}$ is some piecewise continuous function that values $E_{1}$ for $x \le \alpha_{1} + \frac14$, $T_{i}$ on $[\alpha_{i}+\frac34,\alpha_{i+1}+\frac14]$ for $i \in \{1,2,\dots,n-1\}$, $T_{n}$ for $x \ge \alpha_{n} + 3/4$.  It is equivalent to some  neural function.
		\end{lemma}
		
		\begin{proof}
			Observe that any product of the form $\signb{q-\alpha_{i}} \signb{q-\alpha_{j}}$ rewrites to $\signb{q-\alpha_{j}}$ if $\alpha_{i} < \alpha_{j}$ from the definition of $\signb{}$. 
			It follows from induction that above expression rewrites to  $E_{1} + \signb{q-\alpha_{1}} (T_{1}- E_{1}) + \signb{q-\alpha_{2}} (T_{2}- T_{1}) + \dots + \signb{q-\alpha_{n}} (T_{n}- T_{n-1})$, from which the result follows easily.
		\end{proof}

		Actually, let's propose the following notation.
	}
	
	\begin{lemma}\label{lem:switch}
		Assume you are given some integers $\alpha_{1},  \alpha_{2},\dots,$ $\alpha_{n}$, and  some values $V_{1}, V_{2}, \dots,V_{n}$. Then there is some  neural function, that we write 
		$\send({\alpha_{i} \sendsymbol V_{i}})_{i \in \{1,\dots,n\}}$, that maps any $x \in [\alpha_{i}-1/4, \alpha_{i}+1/4]$ to $V_{i}$, for all $i \in \{1,\dots,n\}$. 
	\end{lemma}
	
	\begin{proof}
		Sort the $\alpha_{i}$ so that $\alpha_{1} < \alpha_{2} <\dots, \alpha_{n}$. Then consider $T_{1} + \signb{q-\alpha_{1}} (T_{2}- T_{1}) + \signb{q-\alpha_{2}} (T_{3}- E_{3}) + \dots + \signb{q-\alpha_{n-1}} (T_{n}- T_{n-1}).$
	\end{proof}
	More generally:
	\begin{lemma} \label{lem:switch:pairs}
		Let $N$ be some integer. 
		Assume we are given some integers $\alpha_{1},  \alpha_{2},\dots, \alpha_{n}$, and  some values $V_{i,j}$ for $1 \le i \le n$, and $0 \le j < N$. Then there is some  neural function, that we write 
		$\send({(\alpha_{i},j) \sendsymbol V_{i,j}})_{i \in \{1,\dots,n\}, j \in \{0,\dots,N-1\}}$, that maps any $x  \in  [\alpha_{i}-1/4, \alpha_{i}+1/4]$
		and $y \in [j-1/4,j+1/4]$ to $V_{i,j}$, for all $i \in \{1,\dots,n\}$, $j \in \{0,\dots,N-1\}$.
	\end{lemma}
	
	\begin{proof}
		If we define the function by 
		\begin{align*}
		&\send({(\alpha_{i},j) \sendsymbol V_{i,j}})_{i \in \{1,\dots,n\}, j \in \{1,\dots,N\}}(x,y)\\
		&= \send(N\alpha_{i}+j \sendsymbol V_{i,j})_{i \in \{1,\dots,n\}, j \in \{1,\dots,N\}}(Nx+y) 
		\end{align*}
		
		this works when $x=\alpha_{i}$ for some $i$.
		Considering instead 
		$$\send(N\alpha_{i}+j \sendsymbol V_{i,j})_{i \in \{1,\dots,n\}, j \in \{1,\dots,N\}}(N
		\send(\alpha_{i} \sendsymbol \alpha_{i})_{i \in \{1,\dots,n\}}(x)+y)$$ works for any $x  \in  [\alpha_{i}-1/4, \alpha_{i}+1/4]$.

	\end{proof}

	\section{Simulating Turing machines with functions of $\manonclasslightidealsigmoid$ }
	\label{fptimedansmanonclass}
	
	This section is devoted to proving a kind of reverse implication of the  following proposition, whose proof follows by induction 
	from standard arguments exactly as in \cite{BlancBournezMCU22vantardise}.
	
	\begin{proposition}[{\cite{BlancBournezMCU22vantardise}}] \label{prop:mcu:un}
		All functions of $\manonclasslightidealsigmoid$ \manondufutur{as well as  $\manonclasslighttanh$}  are computable (in the sense of computable analysis) in polynomial time.
	\end{proposition}

	\olivier{Cette rédaction largement perfectible}
	
	We now basically go to prove that for any polynomial time computable function over the reals, we can construct some function $\tilde{\tu f} \in \manonclasslightidealsigmoid$ that simulates the computation of $f$. This basically requires us to be able to simulate the computation of a Turing machine using some functions from $\manonclasslightidealsigmoid$. We basically use the same ideas as in \cite{BlancBournezMCU22vantardise}, but with some improvements, as we need to avoid multiplications, and even get neural functions.

	%
	%
	%

	
	%
	%
	%
	
	Consider without loss of generality some Turing machine $M= (Q, \{0,1\}, q_{init}, \delta, F)$ using the  symbols $0,\symboleun,\symboledeux$, where $B=0$ is the blank symbol. The reason of the choice of symbols $\symboleun$ and $\symboledeux$ will be made clear latter.  We assume $Q=\{0,1,\dots,|Q|-1\}$.  Let 
	$$ \dots  l_{-k} l_{-k+1} \dots l_{-1} l_{0} r_0 r_1 \dots r_n .\dots$$ 
	denotes the content of the tape of the Turing machine $M$. In this representation, the head is in front of symbol $r_{0}$, and $l_i, r_{i} \in  \{0,\symboleun,\symboledeux\}$ for all $i$. 
	Such a configuration $C$ can be denoted by $C=(q,l,r)$, where $l,r \in \Sigma^{\omega}$ are (possibly infinite, if we consider that the tape can be seen as a non finite word, in the case there is no blank on it) words over alphabet $\Sigma=\{\symboleun,\symboledeux\}$ and $q \in Q$ denotes  the internal state of $M$.
	
	The idea is that such a configuration $C$ can also be encoded by some element $\encodageconfiguration(C)=(q, \bar l,\bar r) \in \N \times \R^{2}$, by considering 
	$
	\bar r = \sum_{n \geq 0} r_n 4^{-(n+1)}, \bar l = \sum_{n \geq 0} l_{-n} 4^{-(n+1)}.
	$
	
	Basically, in other words,  we encode the configuration of a bi-infinite tape Turing machine $M$ by real numbers using their radix \base{}  encoding, but using only digits $\symboleun$,$\symboledeux$. 
	If we write: $\encodageconfiguration: \Sigma^{\omega} \to \R$ for the function that maps a word $w=w_{0} w_{1} w_{2} \dots$ to
	$\encodagemot(w)= \sum_{n \geq 0} w_n \base^{-(n+1)}$, we can also write
	$\encodageconfiguration(C)=\encodageconfiguration(q,l,r)= (q,\encodagemot(l),\encodagemot(r)).$
	Notice that this lives in $Q \times [0,1]^{2}$. Denoting the image of $\encodagemot: \Sigma^{\omega} \to \R$ by $\Image$, this even lives in $Q \times \Image^{2}$. 
	%
	%
	%
	
	\shortermcu{
		A key point is to observe that }
	\begin{lemma}
		We can construct some  neural function $\bar {Next}$ in $\manonclasslightidealsigmoid$ that simulates one step of $M$, i.e. that computes the $Next$ function sending a configuration $C$ of Turing machine $M$ to the next one.  
	\end{lemma}
	
	\begin{proof}
		
		We can write  $l = l_0 l^\bullet $ and $r = r_0r^\bullet $, where $l_{0}$ and $r_{0}$ are the first letters of $l$ and $r$, and $l^\bullet$ and $r^\bullet$  corresponding to the (possibly infinite) word  $l_{-1} l_{-2} \dots$ and $r_{1} r_{2} \dots$ respectively.
		%
		%
		\vspace{-0.2cm}
		\begin{center}
			\begin{tabular}{c c|c|c|c c}
				\hline 
				... & $l^\bullet $ & $l_0 $ & $r_0$ & $ r^\bullet$ & ... \\ 
				\hline 
				\multicolumn{1}{c}{} & 
				\multicolumn{2}{@{}c@{}}{$\underbrace{\hspace*{\dimexpr6\tabcolsep+3\arrayrulewidth}\hphantom{012}}_{l}$} & 
				\multicolumn{2}{@{}c@{}}{$\underbrace{\hspace*{\dimexpr6\tabcolsep+3\arrayrulewidth}\hphantom{3}}_{r}$}
			\end{tabular} 	
		\end{center}
		
		%
		%
		The function $ {Next}$ is basically of the form

		$\mathit{Next}(q,l,r) = \mathit{Next}(q,l^\bullet l_0,r_0r^\bullet) = (q', l', r')$ defined as a definition by case of type:
		\[ (q', l', r') = 
		\left\{ \begin{array}{ll}
		(q', l^\bullet l_0 x, r^\bullet) & \mbox{whenever $\delta(q,r_{0}) = (q',x, \rightarrow)$} \\
		(q', l^\bullet, l_0 x r^\bullet) & \mbox{whenever $\delta(q,r_{0}) = (q',x, \leftarrow)$} \\
		\end{array} \right.
		\]
		
		This rewrites as a function $\bar{Next}$ which is similar, working over the representation of the configurations as reals, considering $r_{0} =  \lfloor \base \bar{r}\rfloor$ 
		\begin{align*}
		\mathit{\bar{Next}}(q,\bar l, \bar r) &= \mathit{\bar{Next}}(q,\bar{l^\bullet l_0},\bar{r_0r^\bullet}) = (q', \bar{l'}, \bar{r'})\\
		&= \left\{ 
		\begin{array}{ll}
		(q', \bar{l^\bullet l_0 x}, \bar{r^\bullet}) & \mbox{whenever $\delta(q,r_{0}) = (q',x, \rightarrow)$} \\
		(q', \bar{l^\bullet}, \bar{l_0 x r^\bullet}) & \mbox{whenever $\delta(q,r_{0}) = (q',x, \leftarrow)$} \\
		\end{array}
		\right.
		\end{align*} 
		\vspace{-0.5cm}
		\begin{equation}
		\begin{array}{l} \label{textaremplacer}
		\mbox{
			$\bullet$  in the first case ``$\rightarrow$'' : $\bar{l'} = \base^{-1} \bar l + \base^{-1} x $ and $\bar{r'}  = \bar{r^\bullet} = \{\base \bar r\} $} \\
		\mbox{
			$\bullet$ in the second case ``$\leftarrow$'' : $\bar{l'} =\bar{ l^\bullet} = \{\base \bar l \} $ and $\bar{r'} = \base^{-2} \{\base \bar r\}  + \base^{-1} x + \lfloor \base \bar{l}\rfloor $ }
		\end{array}
		\end{equation}
		
		We introduce the following functions: $\rightarrow : Q \times \{0,1,3\} \mapsto \{0,1\}$ and $\leftarrow:Q \times \{0,1,3\} \mapsto \{0,1\}$ such that $\rightarrow(q,a)$ (respectively: $\leftarrow(q,a)$)
		is $1$ when $\delta(q,a) = (\_,\_, \rightarrow)$ (resp. $(\_,\_, \leftarrow)$), i.e. the head moves right (resp. left), and $0$ otherwise.
		\manon{commenté: 
			\[ \rightarrow(q,a) = \left\{ 
			\begin{array}{ll}
			1 & \mbox{if $\delta(q,a) = (\_,\_, \rightarrow) \in \Delta$} \\
			0 & \mbox{else}
			\end{array}
			\right.
			\]
			
			and 
			
			\[ \leftarrow(q,a) = \left\{ 
			\begin{array}{ll}
			1 & \mbox{if $\delta(q,a) = (\_,\_, \leftarrow) \in \Delta$} \\
			0 & \mbox{else}
			\end{array}
			\right.
			\]

			For $D = \{\rightarrow, \leftarrow \} $, }
		
		We can rewrite $ \mathit{\bar{Next}}(q,\bar l, \bar r) = (q',\bar l',\bar r')$ as:
		$$\bar{l'} = \displaystyle \sum_{q,r_{0}} \left[\rightarrow(q, r_{0}) \left( \frac{\bar{l}}{4} + \frac{x}{4} \right) + \leftarrow(q, r_{0}) \left\{ 4 \bar{l}\right\}\right] $$
		and
		$$\bar{r'} = \displaystyle \sum_{q,r_{0}} \left[\rightarrow(q, r_{0}) \left\{ 4 \bar{r} \right\} + \leftarrow(q, r_{0}) \left( \frac{\left\{ 4r_{0} \right\} }{4^2} + \frac{x}{4} + \lfloor 4 \bar{l} \rfloor \right)\right].$$

		\manon{commenté: 
			$$\bar{l'} = \sum_{\left. \begin{array}{ll}
				~~~~~~~~q\in Q \\
				\delta(q,r_0) = (q',x,D)\\
				\end{array}
				\right. } \rightarrow(q, r_0) \left( \frac{\bar{l}}{4} + \frac{x}{4} \right) + \leftarrow(q, r_0) \left\{ 4 \bar{l}\right\} $$
			
			$$\bar{r'} =  \sum_{ \left. \begin{array}{ll}
				~~~~~~~~q\in Q \\
				\delta(q,r_0) = (q',x,D)\\
				\end{array}
				\right. } \rightarrow(q, r_0) \left\{ 4 \bar{r} \right\} + \leftarrow(q, r_0) \left( \frac{\left\{ 4r \right\} }{4^2} + \frac{x}{4} + \lfloor 4 \bar{l} \rfloor \right) $$
		}
		The problem about such expressions is that we cannot expect the integer part and the fractional part function to be in $\manonclasslightidealsigmoid$ (as functions of this class are computable, and hence continuous, unlike the fractional part).   But, a key point is that from our trick of using only symbols $\symboleun$ and $\symboledeux$, we are sure
		that in an expression like $\lfloor \bar r \rfloor$, either it values $0$ (this is the specific case where there remain only blanks in $r$), or that $\base \bar r$ lives in interval
		$[\symboleun,2)$ or in interval $[\symboledeux,4)$. 
		\manon{MCU: That means that we could replace $\left\{\base \bar r \right\}$ by $\sigma(\base \bar r)$ where $\sigma$ is some  (piecewise affine) function obtained by composing in a suitable way the basic functions of $\manonclasslightidealsigmoid$. 
			\olivier{Il y avait une erreur: voir verification-formules-papier.mw sous maple }
			%
			\olivier{Solution $1$ qui serait correcte: Then, considering $i(x)= \If(x-2,3,\If(x,1,0))$. Solution $2$, j'écris sous une autre forme en simplifiant (merci maple)}
			In particular, $i(x)= \signb{x}+2\signb{x-2}$, 
			$\sigma(x)=x-i(x)$, then  $i(\base \bar r)$ would be the same as $\lfloor \base \bar r \rfloor$, and $\sigma(\base \bar r)$ would be the same as $\{ \base \bar r \}$ in our context in above expressions.
			In other words, for  $r_{0}= i(4\bar{r})$ we could replace the paragraph \eqref{textaremplacer} above by:}
		That means that we could replace $\lfloor \base \bar r \rfloor$ by $\sigma(\base \bar r)$ where $\sigma$ is some  (piecewise affine) function obtained by composing in a suitable way the basic functions of $\manonclasslightidealsigmoid$. 
		In particular, $\sigma(x)= \signb{x}+2\signb{x-2}$, 
		$\xi(x)=x-\sigma(x)$, then  $\xi(\base \bar r)$ would be the same as $\left\{ \base \bar r \right\}$, and $\sigma(\base \bar r)$ would be the same as $\lfloor \base \bar r \rfloor$ in our context in above expressions.
		
		In other words, considering  $r_{0}= \sigma(4\bar{r})$ we could replace the above expression by
		$\bar{l'}=$ $\displaystyle \sum_{q,r_{0}} \left[\rightarrow(q, r_{0}) \left( \frac{\bar{l}}{4} + \frac{x}{4} \right) + \leftarrow(q, r_{0}) \xi( 4 \bar{l})\right]$ \\
		and  $\bar{r'} = \displaystyle \sum_{q,r_{0}} \left[\rightarrow(q, r_{0}) \xi( 4 \bar{r} ) + \leftarrow(q, r_{0}) \left( \frac{\xi( 4r )}{4^2} + \frac{x}{4} + \sigma(4 \bar{l}) \right)\right],$
		\manon{commenté: 
			paragraph \eqref{textaremplacer} above by:
			
			\begin{itemize}
				\item  $$\bar{l'} = \sum_{\left. \begin{array}{ll}
					~~~~~~~~q\in Q \\
					\delta(q,r_0) = (q',x,D)\\
					\end{array}
					\right. } \rightarrow(q, r_0) \left( \frac{\bar{l}}{4} + \frac{x}{4} \right) + \leftarrow(q, r_0) \xi(4 \bar{l})$$
				
				\item $$\bar{r'} =  \sum_{ \left. \begin{array}{ll}
					~~~~~~~~q\in Q \\
					\delta(q,r_0) = (q',x,D)\\
					\end{array}
					\right. } \rightarrow(q, r_0)  \xi(4 \bar{r}) + \leftarrow(q, r_0) \left( \frac{ \xi(4r) }{4^2} + \frac{x}{4} +  \sigma(4 \bar{l})  \right) $$
			\end{itemize}
		}
		%
		%
		and get something that would be still work exactly the same, but using only piecewise continuous functions.
		
		%
		
		%
		%
		%
		%

		\olivier{Il y avait des erreurs dans MCU, sur l'ordre. Normalement, ai remis dans bon ordre, en utilisant Lemma \ref{lem:sigmoids}}
		
		We could then write: 
		\olivier{Solution 1 qui serait correcte: 
			$$q'= \If(q-|Q-2|, nextq^{|Q|-1}, \cdots \If(q-2, nextq^{3}, \If(q-1, nextq^{2}, \If(q-0, nextq^{1},nextq^{0}))))$$

			where 
			$$nextq^{q}=\If(v-1,nextq^{q}_{\symboledeux}, \If(v-0,nextq^{q}_{\symboleun},nextq^{q}_{0}))$$
			and}\olivier{
			
			Alternative: 
			Solution $2$, j'écris sous une autre forme en simplifiant (merci maple)
			
			$$q'= \If(3q+v-3|Q-2|-2, nextq_{3}^{|Q|-1}, \cdots \If(3q+v-6, nextq_{0}^{2}, \If(3q+v-5,nextq_{3}^{1}, \If(3q+v-4, nextq_1^{1}, \If(3q+v-3,$$
			$$ nextq_{0}^{1}\If(3q+v-2,nextq_{3}^{0},\If(3q+v-1,nextq_{1}^{0},nextq_{0}^{0}))))$$
		}\olivier{
			Alternative: Solution 3.

			\olivier{miex couper/latexiser, et écrire cette chose}
			$$q'= nextq_{0}^{0} + \signb{3q+v-1} (nextq_{1}^{0}- nextq_{0}^{0}) + \signb{3q+v-2} (nextq_{3}^{0}- nextq_{1}^{0}) +$$
			$$ \signb{3q+v-3} (nextq_{0}^{1}- nextq_{3}^{0})
			+ \signb{3q+v-4} (nextq_{1}^{1}- nextq_{0}^{1}) + \signb{3q+v-5} (nextq_{3}^{1}- nextq_{1}^{1})$$
			$$  + \dots + \signb{3q+v-3|Q-2|-2} (nextq_{3}^{|Q|-1}- nextq_{1}^{|Q|-1})
			$$
			where
			
			Alternative: Solution 4.
		}
		$$q'=\send((q,r) \sendsymbol nextq^{q}_{r})_{q \in Q, r\in \{0,1,3\}}(q,\sigma(4 \bar r)),$$
		using notation of Lemma \ref{lem:switch:pairs}, 
		where 
		$nextq^{q}_{r_{0}}=q'$ if $\delta(q,r_{0})= (q',x,m)$ for $m \in \{\leftarrow,\rightarrow\}$.

		\olivier{on a bien défi
			ni v?}\olivier{Solution 1 qui serait correcte: 
			Similarly, we can write 
			$$\bar{r'}= \If(q-|Q-2|, nextr^{|Q|-1}, \cdots \If(q-2, nextr^{3}, \If(q-1, nextr^{2}, \If(q-0, nextr^{1},nextr^{0}))))$$

			where

			$$nextr^{q}= \If(v-1,nextr^{q}_{\symboledeux}, \If(v-0,nextr^{q}_{\symboleun},nextr^{q}_{0}))$$
		}\olivier{
			Alternative: Solution 3.

			\olivier{miex couper/latexiser, et écrire cette chose}
			$$\bar{r'}= nextr_{0}^{0} + \signb{3q+v-1} (nextr_{1}^{0}- nextr_{0}^{0}) + \signb{3q+v-2} (nextr_{3}^{0}- nextr_{1}^{0}) +$$
			$$ \signb{3q+v-3} (nextr_{0}^{1}- nextr_{3}^{0})
			+ \signb{3q+v-4} (nextr_{1}^{1}- nextr_{0}^{1}) + \signb{3q+v-5} (nextr_{3}^{1}- nextr_{1}^{1})$$
			$$  + \dots + \signb{3q+v-3|Q-2|-2} (nextr_{3}^{|Q|-1}- nextr_{1}^{|Q|-1})
			$$

			Alternative: Solution 4.
		}\olivier{
			and 
			$
			\bar{r'}= \send((q,v) \sendsymbol nextr^{q}_{v})_{q \in Q, v\in \{0,1,3\}}(q,\sigma(4 \bar r)),$ still using notation of Lemma \ref{lem:switch:pairs}, 
			where $nextr^{q}_{v}$  corresponds to the corresponding expression in the item above according to the value of $\delta(q,v)$.
		}
		\olivier{Mieux expliquer cela: il faudrait etre plus explicite}
		
		We could also replace every $\rightarrow(q,r)$ in above expressions for $\bar l'$ and $\bar l'$ by $$\send((q,r) \sendsymbol \rightarrow(q,r))(q,\sigma(4 \bar r))$$ and symmetrically for $\leftarrow(q,r)$.
		However, if we do so, we still might have some multiplications in the above expressions. 
		
		The key is to use Lemma \ref{tricksigmoid}: We can also write the above expressions as
		$$
		\begin{array}{ll}
		\bar{l'} =   \sum_{q,r}& \left[ 
		\signbb{\send((q,r) \sendsymbol \rightarrow(q,r)) (q,\sigma(4 \bar r)) -\frac34 + \frac12 
			\left( \frac{\bar{l}}{4} + \frac{x}{4} \right) \right.} \\
		&+ \left.
		\signbb{\send((q,r) \sendsymbol \leftarrow(q,r)) (q,\sigma(4 \bar r)) -\frac34 + \frac12 
			\xi( 4 \bar{l})\right]} 
		\end{array}
		$$
		$$
		\begin{array}{ll}
		\bar{r'} =  & \sum_{q,r} \left[ 
		\signbb{\send((q,r) \sendsymbol \rightarrow(q,r)) (q,\sigma(4 \bar r)) -\frac34 + \frac12 
			\xi( 4 \bar{r} ) \right.} \\
		&+ \left.
		\signbb{\send((q,r) \sendsymbol \leftarrow(q,r)) (q,\sigma(4 \bar r)) -\frac34 + \frac12 
			(\frac{\xi( 4r )}{4^2} + \frac{x}{4} + \sigma(4 \bar{l}) )\right]} 
		\end{array}
		$$
		%
		%
		%
		%
		%
		%
		%
		%
		%
		\olivier{Plus besoin:
			The fact that these are ideal neural functions follows from Lemma \ref{lem:sigmoids}. 
		}
	\end{proof}
	
	Once we have one step, we can simulate some arbitrary computation of a Turing machine, using some linear length ODE:

	\begin{proposition} \label{prop:deux} 
		Consider some Turing machine $M$ that computes some function $f: \Sigma^{*} \to \Sigma^{*}$ in some time $T(\ell(\omega))$ on input $\omega$.  One can construct some function $\tilde{\tu f}: \N \times \R \to \R$ in $\manonclasslightidealsigmoid$ that does the same, with respect to the previous encoding: we have $\tilde{\tu f}(2^{T(\ell(\omega))},\encodagemot(\omega))$ provides $f(\omega)$. 
	\end{proposition}
	
	\begin{proof}
		The idea is to define the function $\bar {Exec}$ that maps some time $2^{t}$ and some initial configuration $C$ to the configuration  at time $t$.  This can be obtained using some linear length ODE using previous Lemma.
		$$
		\bar {Exec}(0,C) 
		=C \quad and \quad 
		\dderivl{\bar {Exec}}
		(t, C) 
		=\bar{Next}(\bar {Exec}(t,C)) 
		$$
		
		We can then get the value of the computation as $\bar {Exec}(2^{T(\ell(\omega))}, C_{init})$
		on input $\omega$, considering $C_{init}=(q_{0},0,\encodagemot(\omega))$.
		By applying some projection, we get the following function
		$\tilde{\tu f}(x,y)= \projection{3}{3}(\bar {Exec}(x, q_{0},0,y))$ that satisfies the property.	\end{proof}
	
	\section{Towards functions over the reals}
	\label{sec:computablereal}
	
	The purpose of this section is to prove Theorems \ref{th:main:one}, \ref{th:main:one:ex} and \ref{th:main:one:ex:p}. Actually,  the first two are  a special case of Theorem \ref{th:main:one:ex:p}, so we focus on the latter. \shortermcu{
		\subsection{Reverse implication}}
	The reverse implication of Theorem \ref{th:main:one:ex:p} mostly follows from Proposition  \ref{prop:mcu:un} and arguments from computable analysis. 
	
	For the direct implication, the difficulty is that we know from previous sections how to simulate Turing machines working over the Cantor-like set $\Image$,  while we want functions that work directly over the integers and over the reals.  A first key is to be able to convert from integers/reals to representations using only symbols $\symboleun$ and $\symboledeux$, that is to say, to map integers to $\Image$, and $\Image$ to reals as in \cite{BlancBournezMCU22vantardise}.
	However, we need a stronger statement than the one of \cite{BlancBournezMCU22vantardise} to be able to do both the convention and simultaneously some product (but avoiding the use the multiplication in its definition).
	

	\begin{lemma}[{From $\Image$ to $\R$}, and multiplying in parallel]  \label{lem:codage:manon}
		We can construct some  function $\EncodeMul: \N \times [0,1] \times \R \to \R$ in $\manonclasslightidealsigmoid$ that  maps $\encodagemot(\overline{d})$ and $\lambda$ with $\overline{d} \in \{1,3\}^*$  to  real $\lambda d$. It is surjective over the dyadic, in the sense that for any dyadic $d \in \dyadic$, there is some (easily computable) such $\overline{d}$ with $\EncodeMul(2^{\ell({d})},\overline{d},\lambda )=\lambda d$.
	\end{lemma}

	
	\begin{proof}
		Consider the following transformation: Every digit in the binary expansion of $d$  is encoded by a pair of symbols in the radix $4$ encoding of $\overline{d} \in [0,1]$: digit $0$ (respectively: $1$) is encoded by $11$ (resp. $13$) if before the ``decimal'' point in $d$, and digit $0$ (respectively: $1$) is encoded by $31$ (resp. $33$) if after. For example, for $d=101.1$ in base $2$, $\overline{d}=0.13111333$ in base $4$. 
		%
		%
		The transformation from $\overline{d}$ to $\lambda d$ can be done by considering function $F: \R^{3} \to \R^{3}$ 
		given by 
		\olivier{commenté:
			$
			F( \overline{r_1}, \overline{l_2},\lambda) = 
			\left\{
			\begin{array}{ll} (\xi(16 \overline{r_1}), 2 \overline{l_2} + 0,\lambda) &  \mbox{ whenever } \sigma( 16 \overline{r_1})= 
			5\\
			(\xi(16 \overline{r_1}), 2 \overline{l_2} + \lambda,\lambda) &  \mbox{ whenever } \sigma( 16 \overline{r_1})= 
			7\\
			(\xi(16 \overline{r_1}), (\overline{l_2} + 0)/2,\lambda) &  \mbox{ whenever } \sigma( 16 \overline{r_1})= 
			13\\
			(\xi(16 \overline{r_1}), (\overline{l_2} + \lambda )/2,\lambda) &  \mbox{ whenever } \sigma( 16 \overline{r_1})= 
			15
			\end{array}\right.
			$
			A natural candidate for this is an expression such as 
			Etait pas dans bon ordre selon maple. Mis dans l 'ordre:
			$$\If(i(16 \overline{r_1})-13, (\sigma(16 \overline{r_1}),(\overline{l_2} + \lambda)/2,\lambda),
			\If(i(16 \overline{r_1})-7, (\sigma(16 \overline{r_1}),(\overline{l_2} )/2,\lambda),
			\If(i(16 \overline{r_1})-5, (\sigma(16 \overline{r_1}),2 \overline{l_2} +\lambda ,\lambda),
			,(\sigma(16 \overline{r_1}),2 \overline{l_2}  ,\lambda))))$$
			
			Alternative:
		}
		%
		$\send( 
		5 \sendsymbol (\sigma(16 \overline{r_1}),
		2 \sendsymbol \overline{l_2} + 0,\lambda),$
		$7 \sendsymbol (\sigma(16 \overline{r_1}), 2 \overline{l_2} + \lambda,\lambda),$
		$13 \sendsymbol (\sigma(16 \overline{r_1}), (\overline{l_2} + 0)/2,\lambda),$
		$15 \sendsymbol (\sigma(16 \overline{r_1}), (\overline{l_2} + \lambda )/2,\lambda) )
		(16 \overline{r_1})$
		\olivier{ok que plus besoin de i: commenté}
		with  $\sigma$ and $\xi$ 
		constructed as suitable approximation of the integer part and the fractional part, as in the previous section. 
		%
		We then just need to apply $\ell(\overline{d})$ times $F$ on $(\overline{d},0,\lambda)$, and then project on the second component to get a function $\Encode$ that does the job. That is $\EncodeMul(x,y,\lambda)= \projection{2}{3}(G(x,y,\lambda))$
		with 
		$
		G(0,y,\lambda) = (\overline{d},0,\lambda) 
		\quad and \quad 
		\dderivl{G}
		(t, \overline{d},\overline{l},\lambda) 
		=F(G(\overline{d},\overline{l},\lambda)).
		$
	\end{proof}
	In a symmetric way:
	
	\olivier{Y'avait un bug dans Mcu: soit ca part de $\N^{d}$, soit de $\N$, mais il faut choisir.}
	\begin{lemma}[{From $\N$ to $\Image$, \cite{BlancBournezMCU22vantardise}}]  \label{lem:manquant} We can construct some function $\Decode: \N^{d} \to \R$ in $\manonclasslightidealsigmoid$ that maps $n \in \N^{d}$ to some (easily computable) encoding of $n$ in $\Image$.
	\end{lemma}
	
	\olivier{En fait, on prouve le Theorem \ref{th:main:one}, qui est déjà prouvé dans Manon et al. Adapter, ou alors dire que le super truc c'est que plus de multiplication}
	\olivier{
		We can now  prove the direct direction of  Theorem \ref{th:main:one}: Assume that $\tu f: \N^{d} \to \R^{d'}$ is computable in polynomial time. That means that each of its components are, so we can consider without loss of generality that $d'=1$. We assume also that $d=1$ (otherwise consider either multi-tape Turing machines, or some suitable alternative encoding in $\Encode$).  That means that we know that there is a TM polynomial time computable functions $d: \N^{d+1} \to \{\symboleun,\symboledeux\}^{*}$ so that on $\tu m,n$ it provides the encoding of some dyadic  $\phi(\tu m,n)$ with $\|\phi(\tu m,n)-\tu f(\tu m)\| \le 2^{-n}$ for all $\tu m$. 
		
		From Proposition \ref{prop:deux}, we can construct $\tilde{d}$ 
		with $\tilde{d}(2^{p(max(\tu m,n))},\Decode(n,\tu m))=d(\tu m,n)$ for some polynomial $p$ 
		corresponding to the time required to compute  $d$. 
		
		Both functions $\length{\tu x}=\length{x_1}+ \ldots
		+	\length{x_p}$ and $B(\tu x)=2^{\length{\tu x}\cdot \length{\tu x}}$ are in $\linearderivlength$ (see \THEPAPIERS). It is easily seen that : $\length{\tu x}^c\leq B^{(c)}(\length{\tu x}))$ where $B^{(c)}$ is the $c$-fold composition of function $B$.
		
		Then   $\tilde{\tu f}(\tu m,n)=Encode(\tilde{d}( B^{(c)}(\max(\tu m,n)), \Decode(n,\tu m)))$  provides a solution such that
		$\|\tilde{\tu f}(\tu m,2^{n})-\tu f(\tu m)\| \le 2^{-n}.$
	}
	\olivier{ section{Proving Theorems \ref{th:main:one:ex} and \ref{th:main:one:ex:p}}}
	
	\olivier{Enlevé ce texte qui dit que c'était super dur. Mais peut etre en redire un peu... pour nous amis les referees:
		
		This is clearly a harder task. In particular, a natural approach would be to consider some function $\Encode$ from $\R$ to $\Image$. Unfortunately, such a function $Decode$ is necessarily discontinuous, hence not-computable, hence cannot be in  the class:
		Given some real $d \in [0,1]$ (or $\R$), it would map it to $\bar d$ in $\Image$, that is to say some real whose radix $\base$ expansion contains only $\symboleun$ and $\symboledeux$, this would map it somehow to a sequence. Then we could use the previous techniques to simulate some arbitrary computation of a Turing machine, and use the previous $Encode$ function to map back the result of the computation to some real.
		
		The approach of \emph{mixing} of \cite{BCGH07} provides a solution, even if the constructions there, based on (classical) continuous ODEs use deeply some closure properties of these functions that are not true for discrete ODEs.
		We develop here some ideas based on these principles.

		Let $f: \R \mapsto \R$ computable in polynomial time.
		
		Considering $f_{|\N} : \N \mapsto \R$, by Theorem 5 (main Theorem 1) we know that there exists $\tilde{f} \in \manonclasslightidealsigmoid$ such that for all $\bf{x} \in \N, n\in \N, \|\tilde{\tu f}(\tu x,2^{n}) - \tu f_{|\N}(\tu x) \| \le 2^{-n}$. The problem is that we want to be able to talk about real arguments, and not only integer arguments.

		Assume we have such function $\tilde{f}\in {\manonclass}$, then, from definitions, $f$ is computable in polynomial time. Notice that the basic functions that we added are indeed polynomial time computable, and hence we still have polynomial time computability by induction.

	}
	
	%
	%
	%
	
	\olivier{commenté: , thus, there exists an oracle Turing machine $M$ that computes a function $\tilde{f}$ such that for all $n\in \N$ and for all $x\in \R \cap D$, $\| \tilde{f}(x, n) - f(x) \| \leq 2^{-n}$.
		But since $M$ is an \textit{oracle} Turing machine, we cannot use directly what we did previously.
		
		As before, we would like to encode the execution of $M$ into a function in ${\manonclasslightidealsigmoid}$. Let us take $n\in \N$, intuitively the precision we want on the output and $x\in D \subseteq \R$. We have that setting $n$ implies setting $m=m(n)$, the precision required on the input. $m(n)$ is the modulus of continuity of the function, and is polynomial when $f$ is polynomial time computable. This allows to get back to classical Turing machine, and then to our previous discussions: 
	}\olivier{Pas urgent, mais Il y a un truc sale. Normalement c'est plus le cas, mais des fois on écrivait la multiplication $\times$, des fois $.$, des fois $*$. Verifier qu'on fait plus ca, et qu'on utilise que .}
	
	We now go to the proof of the direct implication of Theorem \ref{th:main:one:ex}. By lack of space, we discuss only the case $d=d'=d''=1$, i.e. of a polynomial time computable function $f: \R \times \N \to \R$. The general case is easy to adapt, by adding suitable arguments, and considering multi-tape Turing machines. From standard arguments from computable analysis (see e.g. [Corollary 2.21]\cite{Ko91}), the following holds\footnote{The idea is that $m$ corresponds to the polynomial modulus of continuity of the function, and $\tilde{f}$ gives some approximation of the function, requiring only requests on a real argument $x$ with precision $m(n,M)$.}.
	\begin{lemma}
		Assume $f: \R \times \N \to \R$ is computable in polynomial time. There exists some polynomial $m: \N^{2} \to \N$ and  some $\tilde{f}: \N^{3} \to \Z$ computable in polynomial time
		such that for all $x\in \R$, $\| 2^{-n} \tilde{f}(\lfloor 2^{m(n,M)}  x  \rfloor, u, 2^{M}, 2^{n}) - f(x,u) \| \leq 2^{-n}$ whenever $\frac{x}{2^{m(n,M)}} \in [-2^{M},2^{M}]$.
	\end{lemma}
	
	Assume we consider an approximation $\sigma_{i}$ (with either $i=1$ or $i=2$) of the integer part function given by Lemma \ref{lem:i}. Then, given $n,M$, when $2^{m(n,M)}  x$ falls in some suitable interval $I_{i}$ for $\sigma_{i}$, we are sure that $\sigma_{i}(2^{m(n,M)+X}, 2^{m(n,M)}  x) = \lfloor 2^{m(n,M)}  x  \rfloor$. Consequently, $2^{-n} \tilde{f}(\sigma_{i}(2^{m(n,M)+X}, 2^{m(n,M)}  x ), u$, $2^{M}, 2^{n})$ provides some $2^{-n}$-approximation of $f(x)$. 
	
	Now, we can compute the value of $\tilde{f}$ on these arguments by simulating the Turing machine 
	that computes $f$,   using functions from $\manonclasslightidealsigmoid$.
	Namely, we just need to map all arguments (expected to be integers) to Cantor-like set $\Image$, then use Proposition \ref{prop:deux} to compute the (encoding over $\Image$) of the corresponding  integer, and then maps it back to some integer value. In other words, we can get an approximation of $f(x,u)$ of the form :
	$\EncodeMul (T_{1},\tilde{\tilde{f}}(T_2,	\Decode(T_{3},(\sigma_{i}(2^{m(n,M)+X}, 2^{m(n,M)}  x ), u, 2^{M}, 2^{n}))),$ $2^{-n})$

	\noindent for some suitable $T_{1}$, $T_{2}$ and $T_{3}$ of polynomial size, big enough to cover (up to some constant) the time required by the Turing machine: here $\tilde{\tilde{f}}$ is the function obtained from $\tilde{f}$ by  Proposition \ref{prop:deux}. This works  when $2^{m(n,M)}  x$ falls in the suitable interval $I_{i}$. Setting $T_{1}$, $T_{2}$ and $T_{3}$ can be done exactly as we did previously (e.g. in Proposition \ref{prop:deux}).
	
	The problem is that it might also be the case that $2^{m(n,M)}  x$ falls in the complement of the intervals $I_{i}$. In that case, we have no clear idea of what could be the value of $\sigma_{i}(2^{m(n,M)+X}, 2^{m(n,M)}  x) $, and of what might be the value of the above expression $Formula_{i}(x,u,$ $M,n)$. 
	But the point is that when it happens for an $x$ for $\sigma_{1}$, we could have used $\sigma_{2}$, and this would work, as one can check that the intervals of type $I_{1}$ covers the complements of the intervals of type $I_{2}$ and conversely. They also overlap, but when $x$ is both in some $I_{1}$ and $I_{2}$, $Formula_{1}(x,u,M,n)$ and $Formula_{2}(x,u,M,n)$ may differ, but they are both $2^{-n}$ approximation of $f(x)$.
	
	The key is then to compute some suitable "adaptive" barycenter, using function $\lambda$, provided by Lemma \ref{lem:lambda}. 
	Observe from the statements of Lemma \ref{lem:i}, and from the statement of Lemma \ref{lem:lambda}  
	\shortermcu{	\begin{itemize}
			\item}  that whenever $\lambda(2^n,x) = 0$, we know that $\sigma_{2}(2^n,x) = \lfloor x \rfloor$; 
		\shortermcu{ \item that} whenever $\lambda(2^n,x) = 1$ we know that $\sigma_{1}(2^n,x) = \lfloor x \rfloor$;
		\shortermcu{ \item that}    whenever $\lambda(2^n,x) \in (0,1)$, we know that $\sigma_1(2^n,x) = \lfloor x \rfloor +1 $ and $\sigma_2(2^n,x) = \lfloor x \rfloor$.
		\shortermcu{\end{itemize}}
	That means that if we consider $\lambda(2^n,x) Formula_{1}(x,u,M,n) + (1-\lambda(2^n,n)) Formula_{2}(x,u,M,n) $ we are
	sure to get some $2^{-n}$ approximation of $f(x)$.
	
	There remains that this requires some multiplication with $\lambda$. But from the form of $Formula_{i}(x,u,M,n)$, this could be also be written as follows, and hence remain in $\manonclasslightidealsigmoid$:
	
	\begin{equation}\label{letrucalafin}
	\EncodeMul(T_{1},\tilde{\tilde{f}}(T_2,\Decode(T_{3},(\sigma_{1}(2^{m(n,M)+X}, 2^{m(n,M)}  x ), u, 2^{M}, 2^{n}))),\lambda(2^{n},x) 2^{-n}) +
	\end{equation}
	\begin{align*}
	 \EncodeMul(&T_{1},\tilde{\tilde{f}}(T_2,\Decode(T_{3},(\sigma_{2}(2^{m(n,M)+X}, 2^{m(n,M)}  x ), u, 2^{M}, 2^{n}))),\\
	 &1-\lambda(2^{n},x) 2^{-n}) 
	\end{align*}
	\olivier{Commenté:
		\begin{proof}
			The representation of integers being finite, we do not need, for all $n \in \N$, some arbitrary precision on the input, but only precision $2^{-m}$. Then $\tilde{\tilde{f}}$ is basically the same as function $\tilde{f}$, but replacing the oracle calls by readings on its first argument.
		\end{proof}
	}

	\section{Normal form and some of its consequences}
	\label{sec:generalizations}

	From the proofs we also get a normal form theorem, namely formula \eqref{letrucalafin}. In particular, \begin{theorem}
		Any function $f: \N^{d} \to \R^{d'}$ can be obtained from the class $\manonclasslightidealsigmoidlim$ using only one schema $\MANONlim$ \olivier{(or $\MANONlimd$)}.
	\end{theorem}
	
	We now go the discussion and proof of Theorem \ref{th:nn}: Observing Formula \ref{letrucalafin}, we see that when $n$ and $M$ are fixed, the expression depends on $u$, $T_{1}$, $T_{2}$ and $T_{3}$. From our hypothesis that the function is of type $f: \R^{d} \to \R^{d'}$, we are in the case $d''=0$, i.e, where there is no such $u$.  These $T_{i}$ functions correspond  basically to the (polynomial) time required by the Turing machine to compute the function $f$.  From the previous  constructions, it turns out that when this time is fixed to some polynomial $t(n)$, the function is some neural function: i.e. formula \ref{letrucalafin} is providing some neural function that is guaranteed to be at precision $2^{-n}$ of $f(x)$ over $[-2^{M},2^{M}]$. This corresponds to the statement of Theorem \ref{th:nn}.
	In other works,  Formula \ref{letrucalafin} can be seen as a function that generates uniformly a family of circuits/formal neurons approximating a given function at some given precision over some given domain. 
	
	Notice that we believe that the $\signb{}$ function can actually be replaced by the $\tanh$ function, at the price of a discussion of the involved errors. We leave this as future work.
	
	\olivier{Mettre ici certaines genéralisations prouvées par Manon: mais peut etre pas recopié MCU. Pour l'instant c'est:
		
		Recall that 
		a function $M : \N \rightarrow \N$ is a modulus of convergence of $g: \N \to \R$, with $g(n)$ converging toward $0$   when $n$ goes to $\infty$,  if and only if for all $i>M(n)$, we have $\| g(i)   \| \le 2^{-n} $.
		A function $M :\N \rightarrow \N$ is a uniform modulus of convergence of a sequence $g: \N^{d+1} \to \R$, with $g(\tu m,n)$ converging toward $0$ when $n$ goes to $\infty$  if and only if for all $i>M(n)$, we have $\| g(\tu m,i)  \| \le 2^{-n} $.
		%
		Intuitively, the modulus of convergence gives the speed of convergence of a sequence.
		
		\begin{definition}[Operation $\MANONlimd$] Given $\tilde{\tu f}:\N^{d+1} \to \R \in \manonclass$, $g: \N^{d+1} \to \R$ such that
			for all $\tu m \in \N^{d}$, $n \in \N$,
			$\|\tilde{\tu f}(\tu m,2^{n}) - \tu f(\tu m) \| \le g(\tu m, n)$
			under the condition that $0 \le g(\tu m, n)$ is decreasing to $0$, with $\| g(\tu m,p(n)) \| \le  2^{-n}$ for some polynomial $p(n)$
			then 
			$\MANONlimd(\tilde{\tu f},g)$ is the (clearly uniquely defined) corresponding function  $\tu f: \N^{d} \to \R^{e}$.
		\end{definition}
		
		\begin{theorem}
			We could replace $\MANONlim$ by $\MANONlimd$ in \shortermcu{the statements of} Theorems \ref{th:main:two} and \ref{th:main:twop}.
		\end{theorem}
		
		This is equivalent to prove the following,  and observe from the proof that we can replace in above statement ``$g(\tu m,n)$ going to $0$'' by ``decreasing to $0$'', and last condition by
		$\| g(\tu m,p(n)) \| \le  2^{-n}$.


		\begin{theorem}\label{th:dix} 
			$\tu F: \N^{d} \to \R^{d'}$ is computable in polynomial time iff there exists $\tu f: \N^{d+1} \rightarrow \Q^{d'}$, with $\tu f(\tu m,n)$ computable in polynomial time \unaire{n}, 
			and $g : \N^{d+1} \rightarrow \Q$ such that
			\begin{itemize}
				\item $\| \tu f(\tu m,n) - \tu F(\tu m) \| \leq g(\tu m,n) $
				\item $0  \le g(\tu m,n)$ and $g(\tu m,n)$ converging to $0$ when $n$ goes to $+\infty$,  
				\item with a uniform  polynomial modulus of convergence $p(n)$.%
				%
			\end{itemize}  
		\end{theorem}
	}

	\newpage
	
	\printbibliography
	
	\newpage
	
	\appendix
	
	\section{Extended state of the art}
	
	We provide some complements on the state of the art (Section \ref{stateoftheart}).
	
	As we wrote, our ways of simulating Turing machines have some reminiscence of similar constructions used in other contexts such as Neural Networks \cite{SS95,LivreSiegelmann}. In particular, we use Cantor-like encoding set $\Image$ in a similar way to what is done in these references. These references use some particular sigmoid function $\sigma$ (called sometimes \motnouv{saturated linear function}) that values $0$ when $x \le 0$, $x$ for $0 \le x \le 1$, $1$ for $x \ge 1$. Clearly, this is equivalent to $\signb{\frac14+\frac12x}$, and hence their constructions can be reformulated using the $\signb{}$ function. However, first, the models considered in these references are recurrent, while our constructions are not recurrent neural networks, and second, their models are restricted to live on the compact domain $[0,1]$, which forbids to get functions from $\R \to \R$, while our settings allows more general functions. Our proofs also require functions taking some integer arguments, that would be impossible to consider in their settings (unless at the price of an artificial reencoding).  In some sense, our constructions can be seen as operators that maps to family of neural networks in the spirit of these models, instead of considering a fixed recurrent neural networks.

	While there have been several characterizations of complexity classes over the discrete (see e.g. the monograph \cite{LivreSiegelmann} about the above discussed approach, but not only), but as far as we know, the relation between formal neural networks with classes of computable analysis has never been established before. We believe that our original settings allows to do so, while this is unclear with the above mentioned models.  
	
	Furthermore, once again, we were motivated by (discrete) ordinary differential equations, and the relations to formal neural networks is a side effect, but not the main goal that we wanted to obtain. And our settings is indeed a characterization in terms of classes of discrete ODEs. 
	
	Notice that today's formal neural networks are often built with the so-called \motnouv{ReLU} (that stands for \motnouv{Rectified Linear Unit}) function, that maps $x \le 0$ to $0$, and $x \ge 0$ to $x$. This could be taken as a basis function instead of the function $\signb{}$ by rexpressing the latter with a suitable expression with ReLU's functions. Notice also that our concept of neural function is not assuming that the last layer of the network is made of neurons, and the result may be output by some linear combination of the neurons in the last layer.
	
	If we do not restrict to neural network related models, with respect to all previous contexts, as far as we know, only a few papers have been devoted to characterizations of complexity, and even computability, classes in the sense of computable analysis. There have been some attempts using continuous ODEs \cite{BCGH07}, that we already mentioned, or 
	the so-called $\R$-recursive functions \cite{DBLP:journals/corr/BournezGP16}. For discrete schemata, we only know \cite{brattka1996recursive}  and  \cite{ng2021recursion}, focusing on computability and not complexity.

	\section{Some results from }
	\label{sec:dode}
	
	\subsection{Some general statements}
	
	In order to be as self-contained as possible, we recall in this section some results and concepts from \THEPAPIERSPLUS. All the statements in this section are already present in \THEPAPIERSPLUS: We are just repeating them here in case this helps.  We provide some of the proofs, when they are not in the preliminary ArXiv version.  
	
	As said in the introduction:
	
	\begin{definition}[Discrete Derivative] The discrete derivative of
		$\tu f(x)$ is defined as $\Delta \tu f(x)= \tu f(x+1)-\tu
		f(x)$. We will also write $\tu f^\prime$ for
		$\Delta \tu f(x)$ to help readers not familiar with discrete differences to understand
		statements with respect to their classical continuous counterparts. 
		
	\end{definition}
	
	Several results from classical derivatives generalize
	to the settings of discrete differences: this includes linearity of derivation $(a \cdot
	f(x)+ b \cdot g(x))^\prime = a \cdot f^\prime(x) + b \cdot
	g^\prime(x)$, formulas for products
	and division such as
	$(f(x)\cdot g(x))^\prime =
	f^\prime(x)\cdot g(x+1)+f(x) \cdot g^\prime(x)= f(x+1)  g^\prime(x) +  f^\prime(x)  g(x)$. 
	Notice that, however, there is no simple equivalent of the chain rule, in other words, there is no simple formula for the derivative of the composition of two functions. 
	
	A fundamental concept is the following:

	\begin{definition}[Discrete Integral]
		Given some function $\tu f(x)$, we write $$\dint{a}{b}{\tu f(x)}{x}$$
		as a synonym for $\dint{a}{b}{\tu f(x)}{x}=\sum_{x=a}^{x=b-1}
		\tu f(x)$ with the convention that it takes value $0$ when $a=b$ and
		$\dint{a}{b}{\tu f(x)}{x}=- \dint{b}{a}{\tu f(x)}{x}$ when $a>b$. 
	\end{definition}
	
	The telescope formula yields the so-called Fundamental Theorem of
	Finite Calculus: 
	
	\begin{theorem}[Fundamental Theorem of Finite Calculus]
		Let $\tu F(x)$ be some function.
		Then,
		$$\dint{a}{b}{\tu F^\prime(x)}{x}= \tu F(b)-\tu F(a).$$
	\end{theorem}

	A classical concept in discrete calculus is the one of falling
	power defined as $$x^{\underline{m}}=x\cdot (x-1)\cdot (x-2)\cdots(x-(m-1)).$$
	This notion is  motivated by the fact that it satisfies a derivative formula $
	(x^{\underline{m}})^\prime  = m \cdot x^{\underline{m-1}}$  similar to the classical
	one for powers in the continuous setting.
	In a similar spirit, we 
	introduce the concept of falling exponential. 
	
	
	\begin{definition}[Falling exponential]
		Given some function $\tu U(x)$, the expression $\tu U$ to the
		falling exponential $x$,
		denoted by $\fallingexp{\tu U(x)}$, stands
		for  \begin{eqnarray*}
			\fallingexp{\tu U(x)} &=&
			(1+ \tu U^\prime(x-1)) \cdots
			(1+ \tu U^\prime(1)) \cdot (1+ \tu U^\prime(0))   \\
			&=&
			\prod_{t=0}^{t=x-1} (1+ \tu U^\prime(t)),
		\end{eqnarray*}
		with the convention that $\prod_{0}^{0}=\prod_{0}^{-1}=\tu {id}$, where $\tu
		{id}$ is the identity (sometimes denoted  $1$ hereafter).
	\end{definition}

	This is motivated by the remarks that 
	$2^x=\fallingexp{x}$, and
	that the discrete
	derivative of a falling exponential is given by
	$$\left(\fallingexp{\tu U(x)}\right )^\prime = \tu U^\prime(x) \cdot
	\fallingexp{\tu U(x)}$$
	%
	for all $x \in \N$.
	
	\begin{lemma}[Derivation of an integral with parameters]  \label{derivationintegral}
		Consider $$\tu F(x) = \dint{a(x)}{b(x)} {\tu f(x,t)}{t}.$$
		Then \begin{align*}
			\tu F'(x) = \dint{a(x)}{b(x)}{  \frac{\partial \tu f}{\partial x} (x,t)}{t} &+ \dint{0}{-a^\prime(x)}{\tu f(x+1,a(x+1)+t)}{t}\\
			&+ \dint{0}{b'(x)}{ \tu f(x+1,b(x)+t ) } {t}. 
		\end{align*}
		
		\noindent In particular, when $a(x)=a$ and $b(x)=b$ are constant functions, $\tu F'(x)= $ $ \dint{a}{b}{\frac{\partial \tu f}{\partial
				x} (x,t)}{t},$
		and when $a(x)=a$ and $b(x)=x$,
		$\tu F'(x) = \dint{a}{x}{  \frac{\partial \tu f}{\partial
				x} (x,t)}{t} + \tu f(x+1,x)$.
		
	\end{lemma}

	\begin{proof}
		\begin{eqnarray*}
			\tu F'(x) &=& \tu F(x+1) - \tu F(x) \\
			&=& \sum_{t= a(x+1)}^{b(x+1) - 1} \tu f(x+1,t) -
			\sum_{t= a(x)}^{b(x) - 1} \tu f(x,t)   \\
			&=& \sum_{t= a(x)}^{b(x)-1} \left( \tu  f(x+1,t) - \tu f(x,t)
			\right)    
			+   \sum_{t=a(x+1)} ^{t= a(x)-1} \tu f(x+1,t) 
			+ \sum_{t=b(x)}
			^{b(x+1)-1} \tu f(x+1,t) \\
			&=& \sum_{t= a(x)}^{b(x)-1}  \frac{\partial \tu f}{\partial
				x} (x,t)  + \sum_{t=a(x+1)} ^{t= a(x)-1} \tu f(x+1,t) + \sum_{t=b(x)}
			^{b(x+1)-1} \tu f(x+1,t) \\
			&=& \sum_{t= a(x)}^{b(x)-1}   \frac{\partial \tu f}{\partial
				x} (x,t)  + \sum^{t=-a(x+1)+a(x)-1} _{t=0} \tu f(x+1,a(x+1)+t) \\
			&&
			+ \sum_{t=0}
			^{b(x+1)-b(x)-1} \tu  f(x+1,b(x)+t).
		\end{eqnarray*}
	\end{proof}

	\begin{lemma}[Solution of linear ODE
		] \label{def:solutionexplicitedeuxvariables}
		For matrices $\tu A$ and vectors $\tu B$ and $\tu G$,
		the solution of equation $\tu f^\prime(x,\tu y)= \tu A(\tu f(x,\tu y),\tu h(x, \tu y), x,\tu y) \cdot \tu f(x,\tu y)
		+  \tu
		B (\tu f(x,\tu y), \tu h(x, \tu y),  x,\tu y)$  with initial conditions $\tu f(0,\tu y)= \tu G(\tu y)$ is
		\begin{eqnarray*}\label{soluce}
			\tu f(x,\tu y)  &=&
			\left( \fallingexp{\dint{0}{x}{\tu
					A(\tu f(t,\tu y),\tu h(t, \tu y), t,\tu y)}{t}} \right) \cdot \tu G (\tu
			y)  \\
			&&
			+
			\dint{0}{x}{ \left(
				\fallingexp{\dint{u+1}{x}{\tu A(\tu f(t,\tu y),\tu h(t, \tu y), t,\tu y)}{t}} \right) \cdot
				\tu B(\tu f(u,\tu y),\tu h(u, \tu y), u,\tu y)} {u}.
		\end{eqnarray*}

	\end{lemma}
	
	\begin{proof}
		Denoting the right-hand side by $\tu {rhs}(x,\tu y)$, we  have 
		$$\begin{array}{ccl} \bar {\tu {rhs}}^{\prime}(x,\tu y)&=& \tu A(\tu f(x,\tu y),\tu h(x, \tu y), x,\tu y)  \cdot  \left( \fallingexp{\dint{0}{x}{\tu
				A(\tu f(t,\tu y),\tu h(t, \tu y),  t,\tu y)}{t}} \right)   \cdot \tu G (\tu
		y) \\
		&& +
		\dint{0}{x}{ \left(
			\fallingexp{\dint{u+1}{x}{\tu A(\tu f(t,\tu y), \tu h(t, \tu y),  t,\tu y)}{t}} \right)^{\prime} \cdot
			\tu B(\tu f(u,\tu y), \tu h(u, \tu y), u,\tu y)} {u} \\ 
		&& +
		\left(
		\fallingexp{\dint{x+1}{x+1}{\tu A(\tu f(t,\tu y),\tu h(t, \tu y), t,\tu y)}{t}} \right) \cdot
		\tu B(\tu f(x,\tu y), \tu h(x, \tu y), x,\tu y) \\
		&=& \tu A(\tu f(x,\tu y),\tu h(x, \tu y), x,\tu y)  \cdot  \left( \fallingexp{\dint{0}{x}{\tu
				A(\tu f(t,\tu y),\tu h(t, \tu y),  t,\tu y)}{t}} \right)   \cdot \tu G (\tu
		y)  \\ && + \  \tu A(\tu f(x,\tu y), \tu h(x, \tu y), x,\tu y) \cdot \\
		&& \dint{0}{x}{ \left(
			\fallingexp{\dint{u+1}{x}{\tu A(\tu f(t,\tu y),\tu h(t, \tu y), t,\tu y)}{t}} \right)  \tu B(\tu f(u,\tu y), \tu h(u, \tu y), u,\tu y)} {u} \\
		&&
		+ \ \tu B(\tu f(x,\tu y),\tu h(x, \tu y), x,\tu y)
		\\
		&=&  \tu A(\tu f(x,\tu y), \tu h(x, \tu y), x,\tu y) \cdot \tu {rhs}(x,\tu y) 	+  \tu
		B (\tu f(x,\tu y),\tu h(x, \tu y), x,\tu y)
		\end{array} 
		$$
		where we have used linearity of derivation and definition of falling exponential for the first term, and derivation of an integral (Lemma \ref{derivationintegral}) providing the other terms to get the first equality, and then the definition of falling exponential.
		This proves the property by unicity of solutions of a discrete ODE, observing that $\bar {\tu {rhs}}(0,\tu y)=\tu G(\tu y)$.
	\end{proof}
	
	We write also $1$ for the identity. 
	
	\begin{remark} \label{rq:fund}
		Notice that this can  be rewritten  as 
		\begin{equation} \label{eq:rq:fund} 
		\tu f(x,\tu y)=\sum_{u=-1}^{x-1}  \left(
		\prod_{t=u+1}^{x-1} (1+\tu A(\tu f(t,\tu y), \tu h(t, \tu y),  t,\tu y)) \right) \cdot  \tu B(\tu f(u,\tu y), \tu h(u, \tu y), u,\tu y)
		,
		\end{equation}
		with the (not so usual) conventions that for any function $\kappa(\cdot)$,  $\prod_{x}^{x-1} \tu \kappa(x) = 1$ and $\tu
		B(-1,\tu y)=\tu G(\tu y)$.
		Such equivalent expressions both have a clear computational content. They can
		be interpreted as an algorithm unrolling
		the computation of   $\tu f(x+1,\tu y)$ from the computation of  $\tu
		f(x,\tu y), \tu f(x-1,\tu y), \ldots, \tu f(0,\tu y)$.  
	\end{remark}
	
	A fundamental fact is that the derivation with respect to length provides a way to do a kind of change of variables:
	
	\begin{lemma}[Alternative view, case of Length ODEs] \label{fundob}
		Let 
		$f: \N^{p+1}\rightarrow \Z^d$,
		$\lengt:\N^{p+1}\rightarrow \Z$  be some functions and assume that \eqref{lode} holds considering  $\lengt(x,\tu y)=\length{x}$.
		Then $\tu f(x,\tu y)$ is given by 
		$\tu f(x,\tu y)= \tu F(\length{x},\tu y)$
		where $\tu F$ is the solution of initial value problem
		\begin{eqnarray*}
			\tu F(1,\tu y)&=& \tu f(0,\tu y), \\
			\dderiv{\tu F(t,\tu y)}{t} &=&  
			\tu h(\tu F(t, \tu y),2^{t}-1,\tu y).
		\end{eqnarray*}
	\end{lemma}

	\olivier{Preuve de ca ou pas? Si pas preuve, donner ref}
	
	\olivier{
		\subsection{Results for functions over the integers}
		
		The results established in this section were obtained in \THEPAPIERS for functions over the integers, and polynomial time computability was in the classical sense (as only functions over the integers were considered). 
		
		The point is that we need generalization of them in this article, as we consider functions possibly reals, and (polynomial time) computability is in the sense of computable analysis.

		We write $\norm{\cdots}$ for the sup norm: given some matrix $\tu
		A=(A_{i,j})_{1 \le i \le n, 1 \le j \le m}$, 
		$\norm{\tu A}=\max_{i,j}
		A_{i,j}$.
		
		\olivier{Nouvelle version d'Arnaud dans journal.tex}
		
		\begin{lemma}[Fundamental observation] \label{fundamencore}
			Consider the ODE 
			\begin{equation} \label{eq:bc}
			\tu f^\prime(x,\tu y)=  {\tu A} ( \tu f(x,\tu y), \tu h(x,\tu y),
			x,\tu y) \cdot
			\tu f(x,\tu y)
			+   {\tu B} ( \tu f(x,\tu y), \tu h(x,\tu y),
			x,\tu y).
			\end{equation}
			\textbf{over the integers}.

			Assume:
			\begin{enumerate}
				\item The initial condition $\tu G(\tu y) = ^{def}
				\tu f(0, \tu y)$, as well as $\tu h(x,\tu y)$ are polynomial time computable in $x$ and the length of   $\tu y$. 
				\item ${\tu A} ( \tu f(x,\tu y), \tu h (x,\tu y),
				x,\tu y)$ and ${\tu B} ( \tu f(x,\tu y), \tu h(x,\tu y),
				x,\tu y)$ are \polynomial{} expressions essentially constant in $\tu f(x,\tu y)$.
				
				%
			\end{enumerate}
			
			Then, there exists a polynomial $p$ such that $\length{\tu f(x,\tu y)}\leq p(x,\length{\tu y})$ and $\tu f(x, \tu y)$ is polynomial time computable in $x$ and the length
			of $\tu y$. 
		\end{lemma}

	}

	\olivier{Pas envie de dire ca, non?
		
		The previous statements lead to the following:
		
		\begin{lemma}[Key Observation for linear $\lengt$-ODE]~\label{lem:fundamentalobservationlinearlengthODE}
			Assume that $\tu f$ \textbf{over the integers}  is the solution of \eqref{SPLode} and that functions $\tu g, \tu h, \lengt$ and $Jump_\lengt$ are computable in polynomial time. 
			Then, $\tu f$ 
			is computable in polynomial time (\textbf{in the classical sense}).
		\end{lemma}
		
		\begin{proof}
			Using the hypotheses, it comes that the number of elements in $Jump_\lengt$ is polynomial in $\length{x}+\length{y}$. We can
			then replace parameter $x$ and derivation in $\lengt(x,\tu y)$ by a
			parameter $t\leq (\length{x}+\length{y})^c$, for some $c$, and derivation in $t$ by Lemma
			\ref{fundob}.
			
			This leads to an ODE of the form:
			$$
			\tu f’(x,\tu y)= \overline {\tu A} ( \tu f(x,\tu y),
			x,\tu y) \cdot
			\tu f(x,\tu y)
			+   \overline {\tu B} ( \tu f(x,\tu y),
			x,\tu y).
			$$
			by setting 
			\begin{eqnarray*}
				\overline {\tu A} ( \tu f(x,\tu y),
				x,\tu y) &=& {\tu A} ( \tu f(x,\tu y), h(x, \tu y),
				x,\tu y) \\
				\overline {\tu B} ( \tu f(x,\tu y),
				x,\tu y) &=&
				{\tu B} ( \tu f(x,\tu y), h(x, \tu y),
				x,\tu y).
			\end{eqnarray*}
			
			But then Lemma \ref{fundamencore} applies, and we get precisely the
			conclusion, observing that the fact that the
			corresponding matrix $\overline {\tu A}$ and vector $\overline {\tu
				B}$ are essentially constant in $\tu f(x, \tu y)$ guarantees
			hypotheses of Lemma \ref{fundamencore}. 
		\end{proof}
	}
	
	\newpage
	
	\section{On proofs}
	\label{sec:proofs}
	
	We provide here more details on proofs, that we had to put in appendix by lack of space.

	\subsection{Proof of reverse implication of Theorem \ref{th:main:one:ex}}
	
	\begin{proof}
		Assume there exists $\tilde{\tu f}:\R^{d} \times \N^{d''+2} \to \R^{d'} \in \manonclasslightidealsigmoid$ such that
		for all $\tu x \in \R^{d}$, $X \in \N$, $ \tu x \in\left[-2^{X}, 2^{X}\right]$, $\tu m  \in \N^{d''}$, $n \in \N$,
		$\|\tilde{\tu f}(\tu x, \tu m,2^{X},2^{n}) - \tu f(\tu x, \tu m) \| \le 2^{-n}.$
		
		From Proposition  \ref{prop:mcu:un}, we know that $\tilde{\tu f}$ is computable in polynomial time (in the binary length of its arguments).
		Then $\tu f(\tu x,\tu m)$ is computable: indeed, given $\tu x$, $\tu m$ and $n$, we can approximate $\tu f(\tu x, \tu m)$ at precision $2^{-n}$ on $[-2^{X}, 2^{X}]$ as follows: Approximate $\tilde{\tu f}(\tu x, \tu m,2^{X},2^{n+1})$ at precision $2^{-(n+1)}$ by some rational $q$, and output $q$. We will then have
		$$
		\begin{array}{ll}
		\|q-\tu f(\tu x, \tu m)\| &\le \|q-\tilde{\tu f}(\tu x, \tu m,2^{X},2^{n+1}) \| + \|\tilde{\tu f}(\tu x, \tu m,2^{X},2^{n+1})-\tu f(\tu x, \tu m)\| \\ &\le 2^{-(n+1)} + 2^{-(n+1)} \\ &\le 2^{-n}.
		\end{array}
		$$
		All of this is done in polynomial time in $n$ and the size of $\tu m$, and hence we get that $\tu f$ is polynomial time computable from definitions.
	\end{proof}
	
	\subsection{Proof of Theorem \ref{trucchoseth}}
	
	\olivier{est-ce une arnaque?}
	
	\begin{proof}
		We know that a function $\tu f: \R^{d} \to \R^{d'}$ from $\manonclasslightidealsigmoid$ is polynomial time computable by Proposition \ref{prop:mcu:un}. That means that we can approximate it with arbitrary precision, in particular precision $\frac14$ in polynomial time. Given such an approximation $\tu q$, it is easy to determine the corresponding integer part: return (componentwise) the closest integer to $\tu q$.
		
		\olivier{est-ce une arnaque?}
		\olivier{ref exacte à une proposition?}
		Conversely, if we have a function $\tu f: \N^{d} \to \N^{d'}$ that is polynomial time computable, our previous simulations of Turing machines provide a function in  $\manonclasslightidealsigmoid$  that computes it.
		
	\end{proof}

	\subsection{Proof of Lemma \ref{lem:xi}}
	
	Observe that  $\sig(a,b,x)=\signb{\frac14 + (x-a)/(2(b-a))}$ corresponds to the piecewise continuous sigmoid function $\sig(a,b,x)$ given by 
	\[ \sig(a,b,x) = 
	\left\{ \begin{array}{lll}
	0 & \mbox{if $x \leq a$} \\
	\frac{x-a}{b-a} & \mbox{if $a \le x \le b$} \\
	1 & \mbox{if $b \leq x$}\\
	\end{array} \right.
	\]
	
	\begin{proof}
		It is sufficient to construct some function $\xi$ such that for all $n\in \N$ and $x\in [- 2^{n} , 2^{n}]$, whenever $ x \in [\lfloor x \rfloor + \frac{1}{8}, \lfloor x \rfloor + \frac{7}{8}] $  , $\xi(2^n,x) = \{ x -\frac18 \} $. Indeed, 
		then $\xi_{1}(N,x)=\xi(N,x-\frac58) -\frac14$ 
		and $\xi_{2}(N,x)=\xi(N,x+\frac18)$ would be solution.
		
		Actually, if we take $\xi'$ that satisfies the constraint only when $x \ge 0$, and that values $0$ for $x \le 0$, then $\frac34-\xi'(-x)$ would satisfy the constraint when $x \le 0$, but values $3/4$ for $x \ge 0$. So,  $\xi(N,x)=\xi'(N+1,x)-\xi'(N+1,-x)+\frac34-\frac34 \sig(0,\frac18,x)$ would work for all $x$. 
		
		So it remains to construct $\xi'$ such that for all $n\in \N$ and $x\in [0 , 2^{n}]$, whenever $ x \in [\lfloor x \rfloor + \frac{1}{8}, \lfloor x \rfloor + \frac{7}{8}] $  , $\xi'(2^n,x) = \{ x-\frac18 \} $, and $\xi'(N,x)=0$ for $x \le 0$.
		
		It suffices to define $\xi'$  by induction by $\xi'(2^{0},x) = \frac34 \sig(\frac18,\frac78,x)$, and $\xi'(2^{n+1},x)= \xi'(2^{n},F(2^{n},x))$, where
		$F(K,x)= x- K.\sig(K,K+\frac18,x)$.
		
		Let $I$ be $[\lfloor x \rfloor + \frac{1}{8}, \lfloor x \rfloor + \frac{7}{8}] $, $x \in I$, and let us first study the value of $F(2^{n},x)$: 
		\begin{itemize}
			\item If $x \leq 2^{n}$, by definition of $\sig$, $F(2^n,x) = x$, then $ F(2^n,x)\in I$.
			\item The case $2^{n} < x < 2^{n} + \frac18$ cannot happen as we assume $x \in I$.
			\item If $2^{n} + \frac18  \leq x $ then $F(2^n, x) = x-2^n $ and $F(2^n, x) \in [ \lfloor x \rfloor - 2^n + \frac{1}{8}, \lfloor x \rfloor - 2^n + \frac{7}{8} ] $
		\end{itemize}
		
		Now, the property is true by induction. Indeed, it is true for $n=0$ by the expression of $\xi'(2^{0},x)$. We now assume that it is true for some $n\in \N$. We have
		$\xi'(2^{n+1},x)= \xi'(2^{n},F(2^{n},x))$. Thus, by induction hypothesis, $\xi'(2^{n+1},x)= \{ F(2^{n},x) -1/8\}$. 
		
		Now:
		\begin{itemize}
			\item If $x \leq 2^{n}$, by definition of $\sig$, $F(2^n,x) = x$, then $\xi'(2^{n+1},x)= \{ F(2^{n},x) -1/8\}=\{x-1/8\}$
			\item The case $2^{n} < x < 2^{n} + \frac18$ cannot happen with our constraint $x \in I$.
			\item If $2^{n} + \frac18  \leq x $ then $F(2^n, x) = x-2^n $ and $\xi'(2^{n+1},x)= \{ F(2^{n},x) -1/8\}=\{x-2^{n}-1/8\}=\{x\}$.
		\end{itemize}
		
		Thus the property is proved for all $n$, and from above expressions, we get $\xi_{1}$ and $\xi_{2}$. 
		
		A graphical representation of $\xi_{1}$ and $\xi_{2}$ can be found in Figure \ref{fig:xi}.
		
		\olivier{En fait, il faut définir $F_{0}$, pas fait la}
		
		There remain to prove that the function $\xi'$ is in $\manonclasslightidealsigmoid$. Unfortunately, this is not clear from the recursive definition, but this can be written in another way, from which this follows. Indeed, we have from an easy induction that $$\xi'(2^{n},x)= F(2^{0},F(2^{1},F(2^{2}(\dots, F(2^{n-1},x))))),$$ if we define $F(2^{0},x)$ 
		as $F(2^{0},x)=\xi'(2^{0},x) = \frac34 \sig(\frac18,\frac78,x)$.
		
		Then, we can obtain $\xi'(2^{n},x)=H(2^{n-1},2^{n},x)$ with 
		\begin{align*}
		H(2^{0}, 2^n, x) &= F(2^{n-1}, x) \\
		H(2^{t+1}, 2^n, x) &= F(2^{n-1-t}, H(2^{t}, 2^n, x)) \\
		&= H(2^{t}, 2^n, x) - 2^{n-1-t}.\sig(2^{n-1-t},2^{n-1-t}+\frac18, H(2^{t}, 2^n, x))
		\end{align*}
		
		Such a recurrence can be then seen as a linear length ordinary differential equation, in the length of its first argument. It follows that $\xi'$, and hence $\xi_{1}$ and $\xi_{2}$ are in 
		$\manonclasslightidealsigmoid$.
	\end{proof}
	
	\begin{figure}
		\includegraphics[width=6cm]{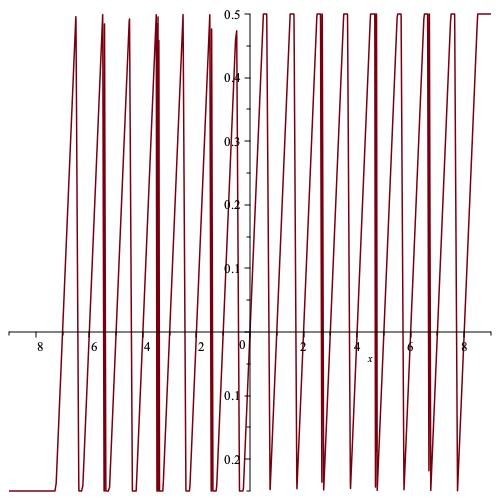} \hfill 
		\includegraphics[width=6cm]{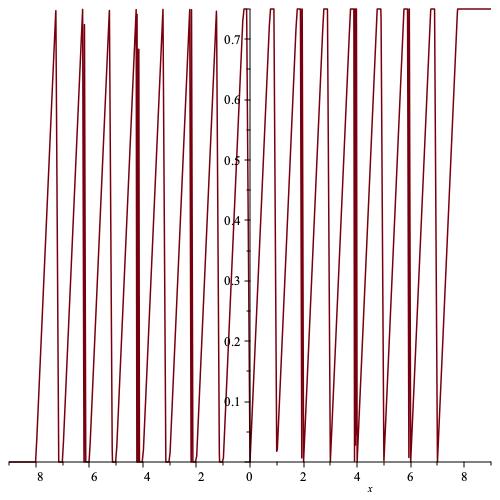}
		\caption{Graphical representation of $\xi_{1}(4,x)$ and $\xi_{2}(4,x)$ obtained with maple. \label{fig:xi}}
	\end{figure}

	\subsection{Proof of Lemma \ref{lem:i}}
	
	\begin{proof}
		Consider $\sigma_i (2^n,x) = x - \xi_i(2^n,x) $ with the function defined in Lemma \ref{lem:xi}.
		
		A graphical representation of $\sigma_{1}$ and $\sigma_{2}$ can be found in Figure \ref{fig:sigma}.
	\end{proof}
	
	\begin{figure}
		\includegraphics[width=6cm]{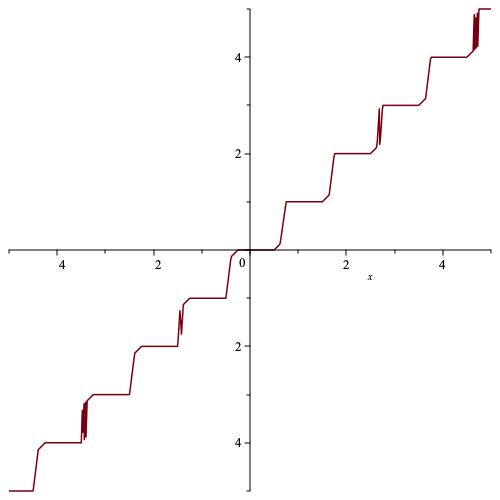} \hfill 
		\includegraphics[width=6cm]{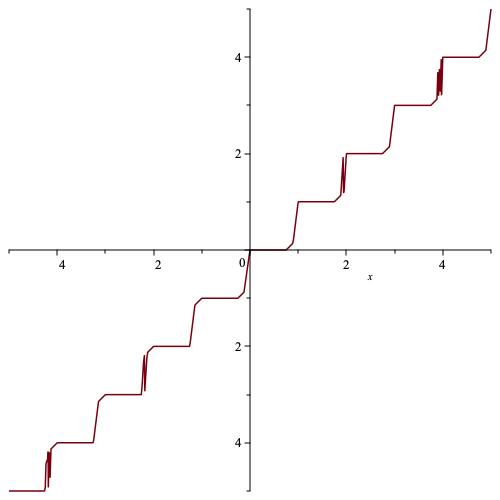}
		\caption{Graphical representation of $\sigma_{1}(4,x)$ and $\sigma_{2}(4,x)$ obtained with maple. \label{fig:sigma}}
	\end{figure}

	\subsection{Proof of Lemma \ref{lem:mod2}}
	
	\begin{proof}
		We can take $\mod_{2}(N,x)=\lambda(N,\frac12x-\frac34)$ where $\lambda$ is the function given by Lemma \ref{lem:lambda}.

		A graphical representation of $\mod{2}$ can be found in Figure \ref{fig:moddeux}.
	\end{proof}
	
	\begin{figure}
		\centerline{
			\includegraphics[width=6cm]{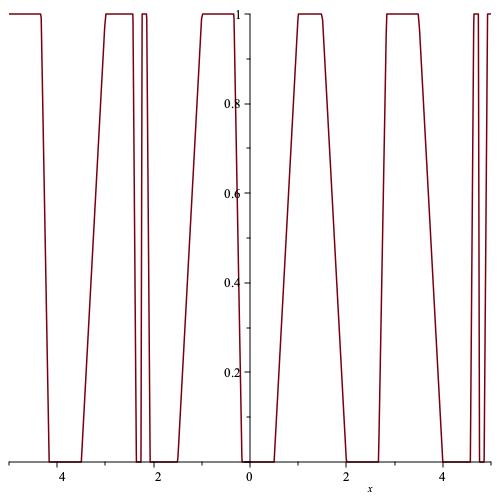} 
		}
		\caption{Graphical representation of $\mod_{2}(4,x)$  obtained with maple. \label{fig:moddeux}}
	\end{figure}

	\subsection{Proof of Lemma \ref{lem:div2}}
	
	\begin{proof}
		We can take $\div{2}(N,x)=\frac12(\sigma_{2}(N,x)-\mod_{2}(N,x))$ where $\mod_{2}$ is the function given by Lemma \ref{lem:mod2}, and $\sigma_{2}$ is the function given by Lemma \ref{lem:i}.
		
		A graphical representation of $\mod{2}$ can be found in Figure \ref{fig:divdeux}.
	\end{proof}
	
	\begin{figure}
		\centerline{
			\includegraphics[width=6cm]{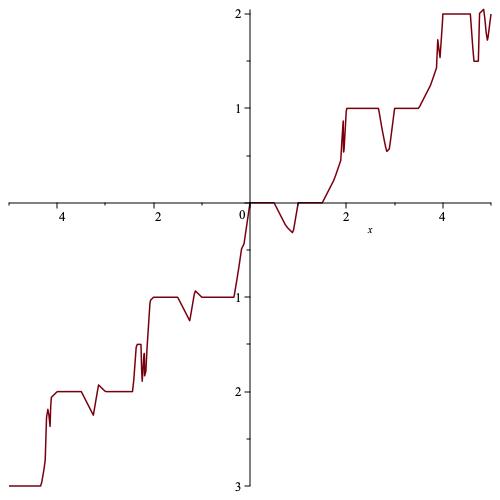} 
		}
		\caption{Graphical representation of $\div_{2}(4,x)$  obtained with maple. \label{fig:divdeux}}
	\end{figure}

	\subsection{Proof of Lemma \ref{lem:lambda}}
	
	The idea is basically to use a technique similar to the one use for Lemma \ref{lem:xi}.
	
	\begin{proof}
		It is sufficient to construct  a function $\lambda'$ that works for $x \ge 0$, and that values $1$ for $x \le 0$. That is to say, such that for all $n\in\N$, $x\in [0 , 2^{n}]$, 
		\shortermcu{\begin{itemize}
				\item} whenever  $ x \in [\lfloor x \rfloor + \frac{1}{4}, \lfloor x \rfloor + \frac{1}{2}] $  , $\lambda'(2^n,x) = 0$, 
			\shortermcu{\item} and whenever $  x \in [\lfloor x \rfloor + \frac{3}{4}, \lfloor x \rfloor +1] $, $\lambda'(2^n,x) = 1$, 
			\shortermcu{\item} and whenever $x \le 0$, then $\lambda'(2^n,x) = 0$.
			\shortermcu{\end{itemize}}
		
		Indeed, then $\lambda(N,x)=\lambda'(N, x) + \lambda'(N, -1/4 - x) - 1$ would be a solution, working for all $x$.
		
		To solve the latter problem, we define $\lambda'$ by induction by $\lambda'(2^{0},x)=\sig(\frac12, \frac34, x) - \sig(0, \frac14, x) + 1$, and
		$\lambda'(2^{n+1},x)= \lambda'(2^{n},G(2^{n},x))$ where 
		$G(N,x)=x-D(N,x)$ and $D(N,x)= N \sig(N-\frac12,N-\frac14,x)$.
		
		By a reasoning similar to the proof of Lemma \ref{lem:xi}, it satisfies by induction the required properties.  
		
		A graphical representation of $\lambda$ can be found in Figure \ref{fig:lambda}.
		
		There remain to prove that the function $\lambda'$ is in $\manonclasslightidealsigmoid$. Unfortunately, this is not clear from the recursive definition, but this can be written in another way, from which this follows. Indeed, we have from an easy induction that  $$\xi'(2^{n},x)= G(2^{0},G(2^{1},G(2^{2}(\dots, G(2^{n-1},x))))),$$ if we define $G(2^{0},x)$ as $G(2^{0},x)=\lambda'(2^{0},x)=\sig(\frac12, \frac34, x) - \sig(0, \frac14, x) + 1.$
		
		Then, we can obtain $\xi'(2^{n},x)=H(2^{n-1},2^{n},x)$ with 
		\begin{align*}
		H(2^{0}, 2^n, x) &= G(2^{n-1}, x) \\
		H(2^{t+1}, 2^n, x) &= G(2^{n-1-t}, H(2^{t}, 2^n, x)) \\
		&= H(2^{t}, 2^n, x) - 2^{n-1-t} \sig(2^{n-1-t}-\frac12,2^{n-1-t}-\frac14,H(2^{t}, 2^n, x)) 
		\end{align*}
		
		Such a recurrence can then be seen as a linear length ordinary differential equation, in the length of its first argument. It follows that $\lambda'$, and hence $\lambda$ are in 
		$\manonclasslightidealsigmoid$.
	\end{proof}
	
	\begin{figure}
		\centerline{
			\includegraphics[width=6cm]{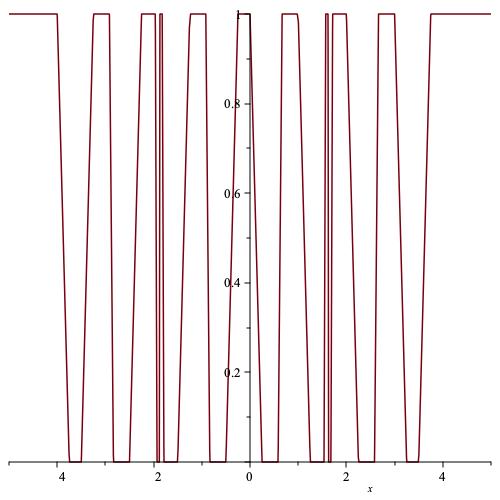} 
		}
		\caption{Graphical representation of $\lambda(4,x)$  obtained with maple. \label{fig:lambda}}
	\end{figure}
	
	\subsection{Proof of Lemma \ref{tricksigmoid}}
	
	\begin{proof}
		Just check that this is true for $d=0$, and then for $d=1$, from the definition of $\signb{}$. 
	\end{proof}
	
	\subsection{Proof of  Proposition \ref{prop:mcu:un}}
	
	We repeat here the idea of the proof of Proposition \ref{prop:mcu:un}, taken from \cite{BlancBournezMCU22vantardise}. 
	
	\begin{proof}
		This is proved by induction.  This is true for basis functions, from basic arguments from computable analysis. In particular as $\signb{.}$ is a continuous piecewise affine function with rational coefficients, it is computable in polynomial time from standard arguments. 
		
		Now, the class of polynomial time computable functions is  preserved by composition: see e.g. \cite{Ko91}: If this helps, the idea of the proof for $COMP(f,g)$, is that by induction hypothesis, there exists $M_f$ and $M_g$ two Turing machines computing in polynomial time $f: \RR \rightarrow \RR$ and $g : \RR \rightarrow \RR$. In order to compute $COMP(f,g)(x)$ with precision $2^{-n}$, we just need to compute $g(x)$ with a precision $2^{-m(n)}$, where $m(n)$ is the polynomial modulus of continuity of $f$. 
		%
		Then, we compute $f(g(x))$, which, by definition of $M_f$ takes a polynomial time in $n$. 
		Thus, since $\mathrm{P}_\RR^{\mathrm{P}_\RR} = \mathrm{P}_\RR$, $COMP(f,g)$ is computable in polynomial time, so the class of polynomial time computable functions is preserved under composition.
		
		It only remains to prove that the class of polynomial time computable functions is preserved by the linear length ODE schema: This is Lemma \ref{lem:un}. 
	\end{proof}

	\olivier{En fait, ca fait partie des preuves assez pourries que j'ai faite. Etre certain que bien vrai \& preuve élégante}
	\begin{lemma}[{\cite{BlancBournezMCU22vantardise}}] \label{lem:un}
		The class of polynomial time computable functions is preserved by the linear length ODE schema.
	\end{lemma}

	\newcommand{\vertiii}[1]{{\left\vert\kern-0.25ex\left\vert\kern-0.25ex\left\vert #1 
			\right\vert\kern-0.25ex\right\vert\kern-0.25ex\right\vert}}
	\newcommand\tnorm[1]{\vertiii{#1}}
	
	
	We write $\MYVEC{x}$ for $2^{x}-1$ for conciseness.
	We write $\tnorm{\cdots}$ for the sup norm of integer part: given some matrix $\tu
	A=(A_{i,j})_{1 \le i \le n, 1 \le j \le m}$, 
	$\tnorm{\tu A}=\max_{i,j}
	\lceil A_{i,j} \rceil $. In particular, given a vector $\tu x$, it can be seen as a matrix with $m=1$, and $\tnorm{\tu x}$ is the sup norm of the integer part of its components.

	\begin{proof} Using   Lemma \ref{fundob}
		(This lemma is repeated from \THEPAPIERS), 
		when 
		\noindent the  schema of Definition \ref{def:linear lengt ODE} holds, 
		we can do a change of variable to consider $\tu f(x,\tu y) $ $ =\tu F(\ell(x),\tu y)$, with $\tu F$ solution of a discrete ODE of the form
		$\dderiv{\tu F(t,\tu y)}{t} = {\tu A} ( \tu F(t,\tu y), \tu h(\MYVEC{t},\tu y),
		\MYVEC{t}, 
		\tu y) \cdot
		\tu F(t,\tu y)
		+   {\tu B} ( \tu F(t,\tu y), \tu h(\MYVEC{t},\tu y),
		\MYVEC{t}, 	  \tu y),$
		that is to say, of the form  \eqref{eq:bcg} below. It then follows from: 
	\end{proof}

	\begin{lemma}[{Fundamental observation, \cite{BlancBournezMCU22vantardise}}] \label{fundamencoreg}
		Consider the ODE 
		\begin{equation} \label{eq:bcg}
		\tu F^\prime(x,\tu y)=  {\tu A} ( \tu F(x,\tu y), \tu h(\MYVEC{x},\tu y),
		\MYVEC{x},
		\tu y) \cdot
		\tu F(x,\tu y)
		+   {\tu B} ( \tu F(x,\tu y), \tu h(\MYVEC{x},\tu y),
		\MYVEC{x},
		\tu y).
		\end{equation}
		Assume:
		1. The initial condition $\tu G(\tu y) = ^{def}
		\tu F(0, \tu y)$, as well as $\tu h(\MYVEC{x},\tu y)$ 
		are polynomial time computable \unaire{x}. 
		
		2.  ${\tu A} ( \tu F(x,\tu y), \tu h (\MYVEC{x},\tu y),
		\MYVEC{x},
		\tu y)$ and ${\tu B} ( \tu F(x,\tu y), \tu h(\MYVEC{x},\tu y),
		\MYVEC{x},
		\tu y)$ are \polynomial{} expressions essentially constant in $\tu F(x,\tu y)$.
		
		%
		
		Then, there exists a polynomial $p$ such that $\length{\tnorm{\tu F(x,\tu y)}}\leq p(x,\length{\tnorm{\tu y}})$ and 
		$\tu F(x,\tu y)$ is polynomial time computable \unaire{x}.
	\end{lemma}
	\begin{proof}
		The fact that there exists a polynomial $p$ such that $\length{\tnorm{\tu f(x,\tu y)}}\leq p(x,\length{\tnorm{\tu y}})$, follows from the fact that we can write some explicit formula for the solution of \eqref{eq:bcg}: This is Lemma \ref{def:solutionexplicitedeuxvariables}, repeated from  \THEPAPIERS. 
		Now,  bounding the size of the right hand side of formula \eqref{eq:rq:fund} provides the statement. 
		
		Now the fact that $\tu F(x,\tu y)$ is polynomial time computable, follows from a reasoning similar to the one of following lemma (the lemma below restricts the form of the recurrence,  
		but the more general recurrence of \eqref{eq:bcg} would basically not lead to any difficulty): The fact that the modulus of continuity of a linear expression of the form of the right hand side of \eqref{eq:bcg} is necessarily affine in its first argument follows from the hypothesis and from previous paragraph, using the fact that $\signb$ has a linear modulus of convergence.
	\end{proof}

	\begin{lemma}[{\cite{BlancBournezMCU22vantardise}}]
		Suppose that function $\tu f: \N \times \R^{d} \to \R^{d'}$ is such that for all $x, \tu y$,  
		\olivier{Reste de MCU: Pas assez général sous cette forme, meme si ca marche pareil. Tenter d'escroquer en disant qu'on manque de place. Est-ce une escroquerie?}
		$$		\tu f(0,\tu y) 
		=\tu g(\tu y) 
		\quad and \quad
		\tu f(x+1,\tu y) 
		= \tu h(\tu f(x,\tu y),x,\tu y)) 
		$$
		for some functions $\tu g: \R^{d} \to \R^{d'}$ and $\tu h: \R^{d'} \times \R \times \R^{d} \to \R^{d'}$
		both computable in polynomial time \unaire{x}.
		%
		Suppose that the modulus $m_{h}$ of continuity of $\tu h$ is affine in its first argument:  
		For all $\tu f,\tu f' \in [-2^{\tu F}, 2^{\tu F}]$, $\tu y \in [-2^{\tu Y}, 2^{\tu Y}]$, 
		$\|\tu f-\tu f'\| \le 2^{-m_{h}(\tu F,\ell(x),\tu Y,n)}$ implies $|\tu h(\tu f,x,\tu y)-\tu h(\tu f',x,\tu y)| \le 2^{-n}$
		with $m_{h}(\tu F, \ell(x),\tu Y,n) = \alpha n + p_{h}(\tu F,\ell(x),\tu Y)$ for some $\alpha$.
		Suppose there exists a polynomial $p$ such that $\length{\tnorm{\tu f(x,\tu y)}}\leq p(x,\length{\tnorm{\tu y}})$.
		
		Then $\tu f(x, \tu y)$ is computable in  polynomial time \unaire{x}.
	\end{lemma}

	%
	%
	%
	%
	%
	%
	%
	%
	%
	%
	%
	%
	%

	\begin{proof}
		The point is that we can compute  $\tu f(n,\tu y)$ by
		\shortermcu{
			\begin{align*}
			\tu f(n, \tu y) &= \tu h(\tu f(n-1,\tu y), n-1, \tu y) \\
			&= \tu h(\tu h(\tu f(n-2,\tu y), n-2, \tu y), n-1, \tu y) \\
			&= \dots \\
			&= \underbrace{\tu h(\tu h(\dots \tu h}_{n}(\underbrace{f(0,\tu y)}_{g(\tu y)}, 0, \tu y)\dots), n-1, \tu y)
			\end{align*}
			
			Basically, the strategy is to compute 
		}
		$\tu z_{0}=\tu f(0,\tu y)=\tu g(\tu y)$,
		then $\tu z_{1}=\tu f(1,\tu y) = \tu h(\tu z_{0}, 0, \tu y)$, 
		then $\tu z_{2}=\tu f(2,\tu y) = \tu h(\tu z_{1}, 1, \tu y)$, 
		then \dots,
		then $$\tu z_{m}=\tu f(m,\tu y) = \tu h(\tu z_{m-1}, m-1, \tu y)$$.
		One needs to do so with some sufficient precision so that the result given by $\tu f(l,\tu y)$ is correct, and so that the whole computation can be done in polynomial time. 
		
		Given $\tu y$, we can determine $\tu Y$ such that $\tu y \in [-2^{\tu Y},2^{\tu Y}]$. 
		Assume for now that for all $m$,
		\begin{equation}
		\label{eq:bienborne}
		z_{m} \in [-2^{Z_{m}},2^{Z_m}]
		\end{equation}
		
		\olivier{Reste de MCU: Version top-down: pour aider à comprendre.
			
			\begin{itemize}
				\item 
				It is basically sufficient to determine $z_{l}=\tu f(l,\tu y)=\tu h(z_{l-1},l-1,\tu y)$ with precision $2^{-n}$.  (*)
				
				\item To get such an approximation (*), it suffices to approximate $z_{l-1}$ 
				with precision $2^{-m_{h}(Z_{l-1},\ell(l-1),Y,n)}$ (**). 
				Then
				indeed, $z_{l}$ could then be computed in a time $poly(Z_{l-1},\ell(l-1), Y, n )$.  
				
				\item To get such an approximation (**) of  $z_{l-1}=\tu f(l-1,\tu y)=\tu h(z_{l-2},l-2,\tu y)$ 
				, it suffices to approximate $z_{l-2}$
				with precision $2^{-m_{h}(Z_{l-2},\ell(l-2),Y,m_{h}(Z_{l-1},\ell(l-1),Y,n))}$ 
				
				We have: 
				$$m_{h}(Z_{l-2},\ell(l-2),Y,m_{h}(Z_{l-1},\ell(l-1),Y,n)) = \alpha^{2} n + \alpha p_{h}(Z_{l-1},\ell(l-1),Y) + p_{h}(Z_{l-2},\ell(l-2),Y)$$

				Then
				indeed, $z_{l-1}$ could then be computed in a time $poly(Z_{l-2},\ell(l-2), Y, m_{h}(Z_{l-1},\ell(l-1),Y,n) )$.  
				\item and so on, until $z_{0}$

			\end{itemize}
		}

		\olivier{Version bottom-up}
		
		For $i=0,1,\dots l$, consider 
		$p(i)= \alpha^{l-i} n + \sum_{k=i}^{l-1} \alpha^{k-i} p_{h}(\tu Z_{k},\ell(k),\tu Y).$
		
		Using the fact that $\tu g$ is computable, approximate $\tu z_{0}=\tu g(\tu y)$ with precision $2^{-p(0)}$. This is doable polynomial time \unaire{p(0)}.
		
		Then for $i=0,1,\dots, l$, using the approximation of $\tu z_{i}$ with  precision $2^{-p(i)}$, compute an approximation of $\tu z_{i+1}$ with precision $2^{-p(i+1)}$: this is feasible to get precision $2^{-p(i+1)}$ of $\tu z_{i+1}$, as $\tu z_{i+1}=\tu f(i+1,\tu y) = \tu h(\tu z_{i},i,\tu y)$, 
		it is sufficient to consider precision
		\begin{align*}
		m_{h}(\tu Z_i,\ell(i),\tu Y,p(i+1)) &=\alpha p(i+1) + p_{h}(\tu Z_{i},\ell(i),\tu Y)\\ 
		&= \alpha^{l-i-1+1} n +  \sum_{k=i+1}^{l-1} \alpha^{k-i-1+1} p_{h}(\tu Z_{k},\ell(k),\tu Y)
		+ p_{h}(\tu Z_{i},\ell(i),\tu Y) \\
		&= p(i).
		\end{align*} 
		Observing that $p(l)=n$, we get $z_{l}$ with precision $2^{-n}$.
		All of this is is indeed feasible in polynomial time \unaire{l}, under the condition that all the $Z_{i}$ remain of size polynomial, that is to say, that we have
		indeed \eqref{eq:bienborne}. But this follows from our hypothesis on $\length{\tnorm{\tu f(x,\tu y)}}$.

	\end{proof}

	\subsection{Proof of Lemma \ref{lem:codage:manon}}

	We provide more details and intuition on the proof of Lemma \ref{lem:codage:manon}.
	
	To compute $d$, given $\overline{d}$, the intuition is to consider a three-tapes Turing machine $(Q, \Sigma, q_{init}, \delta, F)$  the first tape contains the input ($\overline{d}$), and is read-only, the second and third one are write-only and empty at the beginning. We just use a different encoding on the second tape that the previous one: For the first tape, we do restrict to digits $0,\symboleun,\symboledeux$, while for the second, we use binary encoding.
	
	Writing the natural Turing machine that does the transformation, this would basically do the  following (in terms of real numbers), if we forget the encoding of the internal state.
	
	$$
	F( \overline{r_1}, \overline{l_2},\lambda) = 
	\left\{
	\begin{array}{ll} (\xi(16 \overline{r_1}), 2 \overline{l_2} + 0,\lambda) &  \mbox{ whenever } \sigma( 16 \overline{r_1})= \overline{11}=
	5\\
	(\xi(16 \overline{r_1}), 2 \overline{l_2} + \lambda,\lambda) &  \mbox{ whenever } \sigma( 16 \overline{r_1})= \overline{13}=
	7\\
	(\xi(16 \overline{r_1}), (\overline{l_2} + 0)/2,\lambda) &  \mbox{ whenever } \sigma( 16 \overline{r_1})= \overline{31}=
	13\\
	(\xi(16 \overline{r_1}), (\overline{l_2} + \lambda )/2,\lambda) &  \mbox{ whenever } \sigma( 16 \overline{r_1})= \overline{33}=
	15
	\end{array}\right.
	$$
	Here we write $\overline{ab}$ for the integer whose base $\base$ radix expansion is $ab$. 
	
	This is how we got the function $F$ considered in the main part of the paper. Then the previous reasoning applies on the iterations of function $F$ that would provide some encoding function.
	
	Concerning the missing details on the choice of function $\sigma$ and $\xi$. 
	From the fact that we have only $\symboleun$ and $\symboledeux$ in $\overline{r}$,  the reasoning is valid as soon as $16 \overline{r}$ is in the neighborhood of $16 \overline{r} \in \{\overline{11},\overline{13},\overline{31},\overline{33}\}$. 
	
	\subsection{Proof of Lemma 
		\ref{lem:manquant}}

	\begin{proof}
		We discuss only the case $d=1$. The generalization to the general case is easy to obtain. 
		
		\olivier{Faire référence aux fonctions explicites, pas à cet argument foireux
			
			Let $div_{2}$ (respectively: $mod_{2}$)  denote integer (resp. remainder of) division by $2$: As these functions are from $\N \to \N$, from
			Theorem \ref{th:ptime characterization 2} from \THEPAPIERS{}, they belongs to $\linearderivlength$.  Their expression in $\linearderivlength$, replacing $\sign$ by $\signb$, provides some extensions $\overline{div_{2}}$ and $\overline{mod{2}}$ in $\manonclasslightidealsigmoid$.
		}
		We then do something similar as in  Lemma \ref{lem:codage:manon} but now with  function \\
		$\send(0 \sendsymbol ({div_{2}}(\overline{r_1}), (\overline{l_2} + 0)/2), \quad 1 \sendsymbol ({div_{2}}(\overline{r_1}), (\overline{l_2} + 1)/2)) ({mod_{2}}( \overline{r_1})).
		$
		%
	\end{proof}

	\olivier{comme le théorème est comenté...
		\subsection{Proof of Theorem \ref{th:dix}}

		\begin{proof}
			$\Rightarrow$ If we assume that $\tu F$ is computable in polynomial time, we set $g(\tu m, n) = 2^{-n}$, and we take the identity as uniform modulus of convergence. 
			
			$\Leftarrow$ Given $\tu m$ and $n$, approximate $\tu f(\tu m,p(n+1)) \in \Q$ at precision $2^{-(n+1)}$ by some dyadic rational $q$ and output $q$. This can be done in polynomial time \unaire{n}. 
			We will then have 
			$$\begin{array}{lll}
			\| q - \tu F(\tu m) \| & \le & \| q- \tu f(\tu m,p(n+1)) \| + \| \tu f(\tu m,p(n+1)) - \tu F(\tu m) \| \\
			& \le &  2^{-(n+1)}+ g(\tu m,p(n+1)) \\
			& \le &  2^{-(n+1)}+  2^{-(n+1)} \le 2^{-n}.$$
		\end{array}$$
		%
		
		%
		%
		%
		%
		%
		%
		%
		%
	\end{proof}
	
}

\manondufutur{
	\olivier{La, il manque la preuve du Theorème  \ref{th:main:one:ex:pp}: lié à tanh}
}


\olivier{Not à moi-même. Mis end document là, car chantier ce qui suit, mais existe}

\newpage
\olivier{Begin Partie ``chantier manon'': mis end document, mais ca existe dans les sources}
\end{document}